\documentclass[11pt]{amsproc}
\usepackage{amsmath,amssymb,amsthm,epsfig,color,bm}
\usepackage{hyperref}
\newtheorem{theorem}{Theorem}

\newtheorem{remark}[theorem]{Remark}
\newtheorem{lemma}[theorem]{Lemma}
\newtheorem{proposition}[theorem]{Proposition}

\newtheorem{corollary}[theorem]{Corollary}

\newcommand{\re}{\mathop{\mathrm{Re}}\nolimits}
\numberwithin{equation}{section}
\numberwithin{theorem}{section}
\newcommand{\be}{\begin{equation}}
\newcommand{\ee}{\end{equation}}
\newcommand{\bs}{\begin{split}}
\newcommand{\es}{\end{split}}
\def\pmtwo#1#2#3#4{\left( \begin{array}{cc}#1&#2\\#3&#4\end{array}\right)}
\newcommand{\vertiii}[1]{{\left\vert\kern-0.25ex\left\vert\kern-0.25ex\left\vert #1 
    \right\vert\kern-0.25ex\right\vert\kern-0.25ex\right\vert}}

\title{Asymptotics for a class of planar orthogonal polynomials and truncated unitary matrices}

\author{Alfredo Dea\~{n}o\address{alfredo.deanho@uc3m.es. Departamento de Matem\'aticas, Universidad Carlos III de Madrid, Spain.}, Kenneth T-R McLaughlin\address{kmclaughlin@tulane.edu. School of Science \& Engineering, Tulane University, United States.}, Leslie Molag\address{lmolag@math.uc3m.es. Department of Mathematics, Carlos III University of Madrid, Spain.} and Nick Simm\address{N.J.Simm@sussex.ac.uk. School of Mathematical and Physical Sciences, the University of Sussex, United Kingdom.}}

\begin{document}
\maketitle
\begin{abstract}
We carry out the asymptotic analysis as $n \to \infty$ of a class of orthogonal polynomials $p_{n}(z)$ of degree $n$, defined with respect to the planar measure 
\begin{equation*}
d\mu(z) = (1-|z|^{2})^{\alpha-1}|z-x|^{\gamma}\mathbf{1}_{|z| < 1}d^{2}z, 
\end{equation*}
where $d^{2}z$ is the two dimensional area measure, $\alpha$ is a parameter that can grow with $n$, while $\gamma>-2$ and $x>0$ are fixed. This measure arises naturally in the study of characteristic polynomials of non-Hermitian ensembles and generalises the example of a Gaussian weight that was recently studied by several authors. We obtain asymptotics in all regions of the complex plane and via an appropriate differential identity, we obtain the asymptotic expansion of the partition function. The main approach is to convert the planar orthogonality to one defined on suitable contours in the complex plane. Then the asymptotic analysis is performed using the Deift-Zhou steepest descent method for the associated Riemann-Hilbert problem.\\


\end{abstract}

\section{Introduction and main results}
In this paper, we study the asymptotics of orthogonal polynomials with respect to the measure
\begin{equation}
d\mu(z) = (1-|z|^{2})^{\alpha-1}|z-x|^{\gamma}\mathbf{1}_{|z| < 1}d^{2}z \label{trunc-meas}
\end{equation}
where $d^{2}z$ is the two dimensional area measure and $\alpha>0$, $\gamma>-2$ are parameters and $x>0$. By standard theory, \eqref{trunc-meas} determines a sequence of planar orthogonal polynomials $p_{n}(z)$ of degree $n=0,1,2,3,\ldots$ satisfying
\begin{equation}
\int_{\mathbb{D}}p_{j}(z)\overline{p_{k}(z)}d\mu(z) = \delta_{j,k}, \label{planar-op}
\end{equation}
where $\mathbb{D} = \{z \in \mathbb{C} : |z| < 1\}$ is the unit disc. The class of measures \eqref{trunc-meas} occurs naturally in random matrix theory, as follows. Let
\begin{equation}
U_{N} = \begin{pmatrix} B_{n} & *\\ * & * \end{pmatrix}, \label{uconst}
\end{equation}
where $U_{N}$ is an $N \times N$ Haar distributed unitary matrix, and the principal sub-block $B_{n}$ is of size $n \times n$ with $n < N$. The eigenvalues of $B_{n}$ belong to the closed unit disc. It is known from \cite{ZS00} that the eigenvalues $z_{1},\ldots,z_{n}$ of $B_{n}$ have the following joint probability density function
\begin{equation}
P(z_1,\ldots,z_n) = \frac{1}{Z_{N,n}}\prod_{j=1}^{n}(1-|z_j|^{2})^{\alpha-1}\mathbf{1}_{|z_{j}|\leq1}|\Delta(\bm z)|^{2} \label{jpdf-trunc}
\end{equation} 
where $Z_{N,n}$ is the normalization constant, $\alpha = N-n$, and
\begin{equation}
\Delta(\bm z) = \prod_{1 \leq i < j \leq n}(z_{j}-z_{i}).
\end{equation}
In random matrix theory, the moments of the characteristic polynomial are of interest and have been studied in different contexts. For example, in the purely unitary case $n=N$, they are used to model the value distribution of the Riemann-zeta function on the critical line \cite{Keating_Snaith_2000}. In the general non-Hermitian setting, characteristic polynomials of various ensembles are expected to be related to the two-dimensional Gaussian free field and Gaussian multiplicative chaos, see \cite{WW19} for a general overview of this connection.

The moments are related to the leading coefficients of the orthogonal polynomials as
\begin{equation}
\begin{split}
R_\gamma(x):=\mathbb{E}\left(|\det(B_{n}-x)|^{\gamma}\right)=\frac{n!}{Z_{N,n}}\,\prod_{j=0}^{n-1}\chi_{j}^{-2} \label{momsgin}
\end{split}
\end{equation}
where $p_{j}(z) = \chi_{j}z^{j}+O(z^{j-1})$ is the degree $j$ orthonormal polynomial defined by \eqref{planar-op}. Hence, asymptotic information on the planar orthogonal polynomials can be used to determine asymptotics of \eqref{momsgin}. The limiting regime we are most interested in here is the following. Let $N=N_{n}$ be a sequence of positive integers such that $\mu := n/N_{n} \to \tilde{\mu} \in (0,1)$ as $n \to \infty$. This is typically known as the \textit{regime of strong non-unitarity} and is the regime we mostly consider in this paper. In this regime, the limiting eigenvalue distribution of $B_{n}$ is supported on a smaller disc of radius $\sqrt{\tilde{\mu}}$ centered at the origin \cite{ZS00}.

The truncations generalise a well-studied ensemble of random matrices known as the complex Ginibre ensemble. This corresponds to replacing the algebraic factors in \eqref{trunc-meas} with a Gaussian weight $e^{-n|z|^{2}}$ and can be realised by rescaling $z \to z\sqrt{n}/\sqrt{N}$ and $x \to x/\sqrt{N}$, letting $N \to \infty$ with a fixed $n$ in \eqref{trunc-meas}. Then the joint probability density function of $B_{n}$ becomes
\begin{equation}
\frac{1}{Z_{n}}\prod_{j=1}^{n}e^{-n|z_{j}|^{2}}|\Delta(\bm z)|^{2} \label{gin-pdf}
\end{equation}
and measure \eqref{trunc-meas} becomes
\begin{equation}
d\mu_{\mathrm{Gin}}(z) = e^{-n|z|^{2}}|z-x|^{\gamma} d^2z, \quad z \in \mathbb{C}.\label{giN-neas}
\end{equation}
The underlying random matrix with eigenvalue distribution \eqref{gin-pdf} is defined by sampling an $n \times n$ matrix $G_{n}$ of independent and identically distributed complex normal random variables with mean $0$ and variance $1/n$. In the limit $n \to \infty$, the eigenvalues of $G_{n}$ are distributed uniformly on the unit disc in the complex plane and this is known as the circular law. 
The asymptotics of the planar orthogonal polynomials with respect to \eqref{giN-neas} were first studied in \cite{BBLM14} when $\gamma$ scales with $N$.  More recently, \cite{LY17} carried out a complete analysis for fixed $\gamma > -2$ and also explored the discontinuous behavior for $\gamma$ approaching $0$.  Planar orthogonal polynomials with respect to a different class of measures with discrete rotational symmetry were studied in \cite{BGM15}, where a logarithmic singularity appears (with  $|\gamma|<2$) after suitable transformations. In both of these works $x$ is assumed to be bounded away from $1$.  This was followed by the work \cite{BRG18} in the critical case $x = 1+\tau/\sqrt{n}$ where $\tau \in \mathbb{R}$ is fixed and $|\gamma|<2$. Beyond these works, there were several recent papers investigating Gaussian measures of type \eqref{gin-pdf} where the factor $|z-x|^{\gamma}$ is referred to as the insertion of a point charge \cite{AKS21,B24,BY21a, byun2025three, BY23}. A spherical counterpart to \eqref{trunc-meas} was studied in \cite{BKSY25, byun2025properties, byun2025orthogonal}. Furthermore, the moments of the characteristic polynomial, averaged with respect to \eqref{gin-pdf} have been studied for $|x|<1$ in \cite{WW19}. One of our contributions is to establish a similar asymptotic expansion for the truncated unitary ensemble.
%\begin{theorem}[Webb and Wong \cite{WW19}]
%\label{th:ww}
%Let $G_{N}$ be an $N \times N$ matrix from the complex Ginibre ensemble. Then for any fixed $\delta>0$, as $N \to \infty$ the following asymptotic formula holds uniformly for $0\leq x<1-\delta$ and uniformly for compact sets of $\mathrm{Re}(\gamma)>-2$,
%\begin{equation}
%\mathbb{E}\left(|\det(G_{N}-x)|^{\gamma}\right) \sim N^{\frac{\gamma^{2}}{8}}e^{\frac{N\gamma}{2}\,(x^{2}-1)}\,\frac{(2\pi)^{\frac{\gamma}{4}}}{G(1+\frac{\gamma}{2})}, \label{webbwongform}
%\end{equation}
%where $G(z)$ is the Barnes G-function.
%\end{theorem}

We first present the main results for the asymptotics of the polynomials $p_{n}(z)$ orthogonal with respect to measure \eqref{trunc-meas}. In contrast to the case of orthogonal polynomials on the real line, the asymptotic analysis of orthogonal polynomials with respect to general planar weights is largely open, see \cite{HW21,H24} for recent progress. These results mainly hold in a small neighbourhood of an associated potential theoretic droplet boundary. For the measure \eqref{trunc-meas} we will obtain the asymptotics on all regions of the complex plane.

We work in the strong asymptotic regime defined by $\mu = \frac{n}{N} \to \tilde{\mu} \in (0,1)$, so that $\alpha = n(\mu^{-1}-1) > 0$ and the limiting support of the eigenvalues is the disc $|z|<\sqrt{\tilde{\mu}}$. Let $\Gamma_1$ be the Jordan curve passing through $z=1$ determined by
\begin{align} \label{eq:levelCurveGamma1}
|z| = \bigg{|}\frac{1-x^{2}}{1-x^{2}z}\bigg{|}^{\tilde{\mu}^{-1}-1},
\end{align}
see Figure \ref{fig:zeros002b}. 
%\begin{figure}[h!]
%\centerline{\includegraphics[scale=1]{Graphics/zeros002c}}
%\caption{bla bla.}
%\label{fig:zeros002c}
%\end{figure}
It is proved in Lemma \ref{eq:GammarJordan} below that $\Gamma_1$ is a Jordan curve containing $[0,1)$ in its interior. 
We denote the interior and exterior of $\Gamma_1$ by $\mathrm{Int}(\Gamma_1)$ and $\mathrm{Ext}(\Gamma_1)$ respectively. Furthermore, let $D_{1}(\delta)$ be the disc of radius $\delta>0$ centered at $1$, and let $U$ be a thin neighbourhood of the contour $\Gamma_1$. It is convenient to state our result in terms of the (rescaled) \textit{monic} polynomials $P_{n}(z) = x^{-n}p_{n}(xz)\chi_{n}^{-1}$. The notation $\mathcal O(n^{-\infty})$ below means that the expression goes to $0$ faster than any negative power of $n$ (in fact exponentially) as $n\to\infty$. 
\begin{theorem} \label{thm:mainPoly}
Let $\gamma\in\mathbb C$ be fixed and assume $\mathrm{Re}(\gamma)>-2$. 
Consider the strong regime where $\mu := \frac{n}{N} \to \tilde{\mu}\in (0,1)$. Fix $\delta>0$ and let $\delta_n=\frac{\delta}n$. Then as $n\to\infty$, uniformly in the regions for $z$ described below, uniformly in any compact subset of $x \in (0,\sqrt{\tilde{\mu}})$, $z \in \mathbb{C}$ and $\mathrm{Re}(\gamma)>-2$, we have
\begin{equation*}
P_{n}(z) = 
\begin{cases} 
\begin{aligned}
&z^{n}\left(\frac{z}{z-1}\right)^{\frac{\gamma}{2}}(1+\mathcal O(n^{-\infty})), & z \in \mathrm{Ext}(\Gamma_1)\setminus U\,\\
&\frac{1}{\Gamma\left(\frac{\gamma}{2}\right)}\left(\frac{1-\tilde{\mu}^{-1}x^{2}}{1-x^{2}}\right)^{\frac{\gamma}{2}-1} \times\\
&\left(\frac{1-x^{2}}{1-x^{2}z}\right)^{\alpha+\frac{\gamma}{2}}n^{\frac{\gamma}{2}-1}\frac{1}{1-z}(1+\mathcal O(n^{-1})), & z \in \mathrm{Int}(\Gamma_1)\setminus U\\
& \exp\left(n \frac{1-\tilde\mu^{-1}x^2}{1-x^2}(z-1)\right)
\frac{\Gamma(\frac\gamma2,n\frac{1-\tilde\mu^{-1}x^2}{1-x^2}(z-1))}{\Gamma(\frac\gamma2)} \times\\
& \left(\frac{1-x^{2}}{1-x^{2}z}\right)^{\alpha+\frac{\gamma}{2}}(z-1)^{-\frac\gamma2}
(1+\mathcal O(n^{-1})), & z \in D_{1}(\delta_n).
\end{aligned} 
\end{cases}
\end{equation*}
When $z \in U\setminus D_{1}(\delta_n)$ the asymptotics are given by the sum of the first two cases above, i.e., uniformly as $n\to\infty$
\begin{equation*}
\begin{split}
& P_{n}(z) = \left(\frac{1-x^{2}}{1-x^{2}z}\right)^{\alpha+\frac{\gamma}{2}}n^{\frac{\gamma}{2}-1}\frac{1}{1-z}\frac{1}{\Gamma\left(\frac{\gamma}{2}\right)}\left(\frac{1-\tilde{\mu}^{-1}x^{2}}{1-x^{2}}\right)^{\frac{\gamma}{2}-1}(1+\mathcal O(n^{-1}))\\
&+z^{n}\left(\frac{z}{z-1}\right)^{\frac{\gamma}{2}}(1+\mathcal O(n^{-\infty})), 
\quad z \in U\setminus D_{1}(\delta_n).
\end{split}
\end{equation*}
\end{theorem}

\begin{figure}[h!]
\centerline{\includegraphics[scale=1]{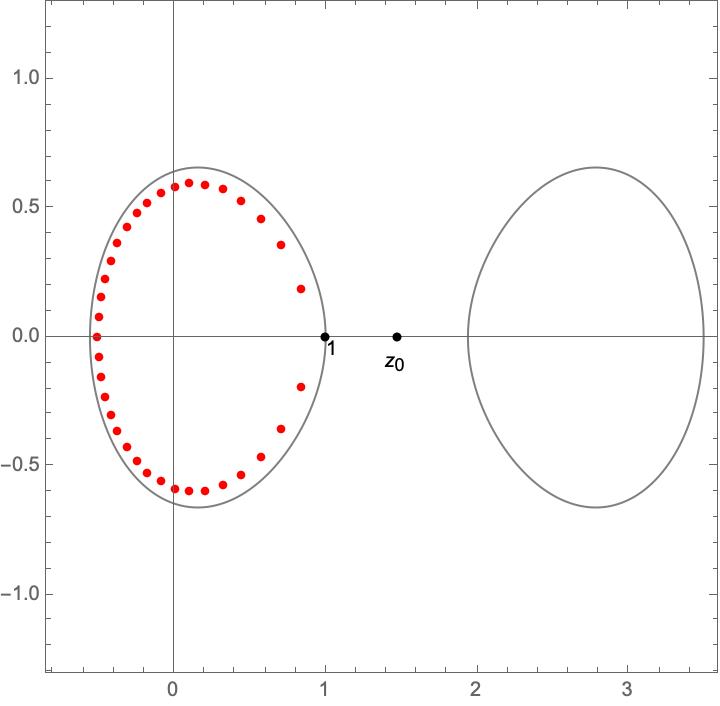}}
\caption{The zeros of the orthogonal polynomials $p_n(x z)$ as defined in \eqref{planar-op}. Here $N=2n=75$, $\gamma=\frac12$ and $x=\frac{7}{12}$. The curve on the left is $\Gamma_1$. The curve on the right is another component of the level curve $\varphi(z):=|z| |1-x^2 z|=1-x^2$ (note that $\tilde\mu=\frac12$). The point $z_0=\tilde\mu x^{-2}=\frac{72}{49}$ is the only point where two level curves of $\varphi$ meet and it is a saddle point.}
\label{fig:zeros002b}
\end{figure}

\begin{remark}
While the estimate in the region $U\setminus D_1(\delta_M)$ suggests a branch cut at $[0,1]$, the estimate for $P_n$ is actually well-defined. For $z<1$ the corresponding jump is cancelled by terms in the error term (see the proof of Proposition \ref{prop:intGammar1}). 
\end{remark}

\begin{remark}
For large $n$, the zeros of $P_n$ must lie in the thin neighborhood $U$ of $\Gamma_1$. See Figure \ref{fig:zeros002b}.
\end{remark}

Of course, one may obtain the behaviors of the original polynomials $p_n(z)$ as defined in \eqref{planar-op} simply by a rescaling $z\to x^{-1}z$ and multiplication by an overall factor, and the asymptotic behaviors, exhibiting singular behavior near the point charge insertion $z=x$, are seen to be comparable to \cite{LY17}[Theorem 3]. For completeness, we mention that the excluded case $x=0$ is trivial since the rotational symmetry gives
\begin{align*}
p_n(z) = \frac1{\sqrt{2\pi B(\alpha,n+\frac\gamma2+1)}} z^n
\end{align*}
and thus $P_n(z)=z^n$.\\

We now present our results on the moments of the characteristic polynomial of matrices sampled from the model of truncated unitary matrices.
   
\begin{theorem}
\label{th:trunc-asympt}
Let $B_{n}$ be the $n \times n$ matrix in \eqref{uconst}. Consider the strong regime where $\mu := \frac{n}{N} \to \tilde{\mu}$ as $n \to \infty$ with $0 < \tilde{\mu}<1$. Then for any fixed $\delta>0$, as $n \to \infty$ the following asymptotic formula holds uniformly for $0 \leq x < \sqrt{\tilde{\mu}}-\delta$ and uniformly for compact subsets of $\mathrm{Re}(\gamma)>-2$,
\begin{equation}
\mathbb{E}(|\det(B_{n}-x)|^{\gamma}) = n^{\frac{\gamma^{2}}{8}}\mu^{\frac{\gamma n}{2}}\left(\frac{1-\mu}{1-x^{2}}\right)^{\frac{\alpha\gamma}{2}}\,C_{\gamma,\mu}(x)\Big(1+\mathcal O\Big(\frac{1}{n}\Big)\Big), \label{rgamx_as}
\end{equation}
where $\alpha=N-n$ and 
\begin{equation*}
C_{\gamma,\mu}(x) = \frac{(2\pi)^{\frac{\gamma}{4}}}{G(1+\frac{\gamma}{2})}\,\left(\frac{\sqrt{1-\mu}}{1-x^{2}}\right)^{\frac{\gamma^{2}}{4}}.
\end{equation*}
\end{theorem}
\begin{remark}
In the particular case $\gamma=2k$ with $k \in \mathbb{N}$, Theorem \ref{th:trunc-asympt} reduces to a particular case of \cite[Theorem 1.3]{SS23}, proved using a special duality that interchanges the roles of $k$ and $n$. The latter approach is limited to integers $k \in \mathbb{N}$.
%In this way, \cite[Theorem 1.3]{SS23} also establishes asymptotics valid uniformly near the critical regime $x = \sqrt{\mu}$. In the Ginibre case of Theorem \ref{th:ww}, asymptotics valid uniformly near the critical regime $x=1$ and post-critical regime $x>1$ were obtained in \cite{DS22} for $\gamma=2k$, $k \in \mathbb{N}$.
\end{remark}
Note that $R_{\gamma}(x)$ is by definition the moment generating function of the random variable $\log|\det(B_{n}-x)|$. As a consequence of the above asymptotics we obtain the following central limit theorem.
\begin{corollary}
\label{cor:clt}
Consider the asymptotic regime described in Theorem \ref{th:trunc-asympt}. Then for any fixed $x \in [0,\sqrt{\tilde\mu})$, as $n \to \infty$ we have the convergence in distribution to a standard normal random variable
\begin{equation}
\frac{\log|\det(B_{n}-x)| - \kappa_{1}(n,x)/2}{\frac{1}{2}\sqrt{\log n}} \overset{d}{\longrightarrow} \mathcal{N}(0,1) \label{clt}
\end{equation}
where
\begin{equation}
\kappa_{1}(M,x) = n\log(\mu)+\alpha\log\left(\frac{1-\mu}{1-x^{2}}\right).
\end{equation}
\end{corollary}
\begin{remark}
This generalises a central limit theorem known to hold for the determinant, i.e. where $x=0$, see \cite[Theorem 4.4]{SS23}.
\end{remark}
The starting point for our approach is related to the one used for the Gaussian weight \eqref{giN-neas} in \cite{BBLM14,BGM15, LY17, lee2023strong, WW19, BKP23}. It involves re-writing the planar orthogonality in terms of suitable orthogonality on contours. This can then be characterised in terms of a Riemann-Hilbert problem, to which the Deift-Zhou steepest descent analysis can be applied. There are some notable differences in the way we carry out the asymptotic analysis compared to \cite{WW19,LY17}. In their construction, the steepest descent contour passes through a branch point at $z=1$ and a local parametrix is constructed around this point which is solved in terms of incomplete gamma functions. For values of $\gamma$ within a given interval $\gamma \in (2k,2(k+1)]$ where $k>-1$ is an integer, the local parametrix has to be suitably corrected depending on the value of $k$. This choice is natural as it turns out that the asymptotic behavior of the (rescaled) orthogonal polynomials in various regions is essentially dependent on the behavior at $z=1$, but it does cause some technical difficulties primarily in the construction of the global and local parametrices. Our approach instead passes the contour through a point away from the branch point and does not require the construction of a local parametrix. Utilizing the specific structure of the jumps of the final transformation, we use a small norm argument that allows us to express the orthogonal polynomials as a contour integral in those regions (see Section \ref{sec:smallNorm}). The integration contour can be deformed (see Section \ref{sec:ContourInt}) to pass through $z=1$ and an argument involving analytic continuation finally allows us to relate this integral to the incomplete gamma function. In this way, we obtain global asymptotics without ever introducing local parametrices. 

The approach described before is for $x$ in the bulk, and it uses the freedom in picking contours to avoid having to deal with the behavior near singular points in the Riemann-Hilbert analysis.  Another interesting regime is the so-called double scaling regime, where $x$ is microscopically close to the droplet boundary, i.e., $x=\sqrt{\tilde\mu}+\mathcal O(1/\sqrt n)$. Here, the freedom in picking contours is lost when $\alpha=N-n\to\infty$, which now are forced to pass through $z=1$. Several steps in our approach have to be adapted, including the construction of a local parametrix near $z=1$, and we postpone this more involved analysis to a future work. However, we are able to investigate the analogous setting in the weak regime where $\alpha=N-n$ is fixed and
\begin{equation}\label{eq:dscaling2}
x^2=1-\frac{v}{n},\qquad v>0.
\end{equation}
This regime has a qualitatively different behaviour to the strong regime in Theorem \ref{th:trunc-asympt}. In this setting, the quantity $R_{\gamma}(x)$ can be written in terms of Toeplitz determinants on the unit circle, with singularity structure matching those considered in \cite{CIK11}, and using that work, we show that the corresponding large $n$ asymptotics involve $\sigma$-Painlev\'e functions.
\begin{theorem}
Consider the function $\sigma(u)$, which is a solution of the Jimbo--Miwa--Okamoto $\sigma$-Painlev\'e V equation:
\begin{equation}\label{eq:PV}
(u\sigma'')^2
=
\left(\sigma-u\sigma'+2(\sigma')^2+2a\sigma'\right)^{2}-4(\sigma')^2(\sigma'+a+b)(\sigma'+a-b),
\end{equation}
with parameters $a=\frac{\alpha+\gamma}{2}$ and $b=\frac{\alpha}{2}$, and which is determined by the following boundary conditions:
\begin{equation}\label{eq:PVbcc}
\sigma(u)
=
\begin{cases}
a^2-b^2+\frac{a^2-b^2}{2a}\left\{u+u^{1+2a}C(a,b)\right\}\left(1+\mathcal{O}(u)\right), & \,\, u\to 0, \,\,2a\notin\mathbb{Z}\\[2mm]
a^2-b^2+\mathcal{O}(u)
+
\mathcal{O}(u^{1+2a})+
\mathcal{O}(u^{1+2a}\log u), & \,\, u\to 0, \,\,2a\in\mathbb{Z}\\[2mm]
\displaystyle -\frac{u^{-1+2a}e^{-u}}{\Gamma(a-b)\Gamma(a+b)}\left(1+\mathcal{O}(u^{-1})\right),&\,\, u\to+\infty.
\end{cases}
\end{equation}
Such a solution of \eqref{eq:PV} exists and is real analytic on $(0, +\infty)$, except for a finite set of possible poles, see \cite[Theorem 1.8]{CIK11}. 

Then, for $\alpha+\frac{\gamma}{2},\frac{\gamma}{2}\neq -1,-2,\ldots$, in the double scaling regime \eqref{eq:dscaling2}, for fixed values of $v\in(0,+\infty)$ bounded away from poles of the function $\sigma$, 
%In the variable $t=-\log x$, for $0 < t < t_0$, with $t_0$ sufficiently small, and 
%
we have the following asymptotic behavior as $n\to\infty$:
%Define $a=\alpha+\gamma$ and $b=\alpha$. \textcolor{blue}{I think that it should be $a=\frac{\alpha+\gamma}{2}$ and $b=\frac{\alpha}{2}$, see (5.4).} Then we have
\begin{align}\label{eq:Rgamma_PV}
R_{\gamma}(x)
=
\left(1+o(1)\right) n^{\frac{\gamma^2}{4}}
v^{-\frac{\gamma}{2}\left(\frac{\gamma}{2}+\alpha\right)}
\exp\left(-\int_{v}^{\infty} \frac{\sigma(u)}{u}du\right).
\end{align}
 %, where $t=-\log x$ and 
\end{theorem}

%We postpone the analysis of the double scaling regime at strong non-unitarity to a future work. 
%\begin{theorem}
%Consider the double scaling regime as described above. Define $a=\alpha+\gamma$ and $b=\alpha$. \textcolor{blue}{I think that it should be $a=\frac{\alpha+\gamma}{2}$ and $b=\frac{\alpha}{2}$, see (5.4).} Then we have
%\begin{align}\label{eq:Rgamma_PV}
%R_{\gamma}(x)
%=
%M^{\frac{\gamma^2}{4}}\left(1+\mathcal{O}(M^{-1})\right)
%v^{-\frac{\gamma}{2}\left(\frac{\gamma}{2}+\alpha\right)}
%\exp\left(-\int_{2Mt}^{\infty} \frac{\sigma(u)}{u}du\right)
%(1+o(1))
%\end{align}
%as $M\to\infty$, where $t=-\log x$ and $\sigma(u)$ is a solution of the Jimbo--Miwa--Okamoto $\sigma$-Painlev\'e V equation:
%\begin{equation*}
%(u\sigma'')^2
%=
%\left(\sigma-u\sigma'+2(\sigma')^2+2a\sigma')\right)-4(\sigma')^2(\sigma'+a+b)(\sigma'+a-b),
%\end{equation*}
%which is real analytic on $(0, +\infty)$ and has specific boundary conditions
%
%\begin{equation*}
%\sigma(u)
%=
%\begin{cases}
%a^2-b^2+\frac{a^2-b^2}{2a}\left\{u+u^{1+2a}C(a,b)\right\}\left(1+\mathcal{O}(u)\right), & \,\, u\to 0, \,\,2a\notin\mathbb{Z}\\[2mm]
%a^2-b^2+\mathcal{O}(u)
%+
%\mathcal{O}(u^{1+2a})+
%\mathcal{O}(u^{1+2a}\log u), & \,\, u\to 0, \,\,2a\in\mathbb{Z}\\[2mm]
%\displaystyle -\frac{u^{-1+2a}e^{-u}}{\Gamma(a-b)\Gamma(a+b)}\left(1+\mathcal{O}(u^{-1})\right),&\,\, u\to+\infty.
%\end{cases}
%\end{equation*}
%\end{theorem}

\section*{Acknowledgements}
A.D. acknowledges financial support from Grant PID2021-123969NB-I00, funded by MCIN/AEI/ 10.13039/501100011033, and from grant PID2021-122154NB-I00 from Spanish Agencia Estatal de Investigaci\'on. K.M. Acknowledges the support of a Royal Society Wolfson Fellowship (grant number: RSWVF\textbackslash R2 \textbackslash 212003).
L.D.M. was funded by the Royal Society grant RF\textbackslash ERE\textbackslash210237 and the internal grant 2024\textbackslash00002\textbackslash007\textbackslash001\textbackslash023 of the Carlos III University of Madrid. 
N.S. acknowledges support from the Royal Society, grant URF\textbackslash R\textbackslash231028.

\section{Formulation as a Riemann-Hilbert problem}
The purpose of this section is to formulate a Riemann-Hilbert problem for the planar orthogonal polynomials defined by \eqref{planar-op}. The first step is a reduction to orthogonality on a suitable contour. Let $\mathcal{C}$ be a positively oriented closed loop encircling $[0,1]$ and avoiding $[1/x^{2},\infty)$, see Figure \ref{fig:curveC}. We will allow that $\mathcal C$ passes through $1$, but not through $0$. Define the weight 
\begin{equation}
w(z) = (1-x^{2}z)^{\alpha+\frac{\gamma}{2}}\left(\frac{z-1}{z}\right)^{\frac{\gamma}{2}}, \qquad z \in \mathcal{C} \label{cweight}
\end{equation} 
where $\alpha=N-n$. 
\begin{figure}[h!]
\centerline{\includegraphics[scale=1.25]{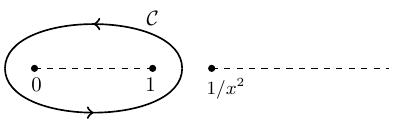}}
\caption{Curve $\mathcal{C}$ and branch cut for the orthogonal polynomials $p_k(z)$.}
\label{fig:curveC}
\end{figure}
Then
\begin{lemma}
\label{lem:planartocontour}
The planar orthogonal polynomials $\tilde{p}_j(z) = x^{-j} p_{j}(x z)$ introduced in \eqref{planar-op} are characterized by the following contour orthogonality 
\begin{equation}
\oint_{\mathcal{C}}\tilde{p}_{j}(z)z^{-k}w(z)\frac{dz}{2iz} = \delta_{j,k}\frac{1}{\hat{\chi}_{j}}, \qquad k=0,\ldots,j. \label{constr}
\end{equation}
where
\begin{equation}
\hat \chi_{j} = \frac{\Gamma(\frac{\gamma}{2}+j+1)\Gamma(\alpha)}{\Gamma(\frac{\gamma}{2}+j+\alpha+1)}\chi_{j}. \label{norm-rel}
\end{equation}
\end{lemma}
\begin{proof}
The proof of this Lemma is essentially contained in \cite[Section 5]{DS22} and follows from an application of Green's theorem, we give the proof in Appendix \ref{sec:planartocontour}.
\end{proof}
The advantage of reducing the planar orthogonality in \eqref{planar-op} to the contour orthogonality \eqref{constr} is that the latter can be formulated in terms of a Riemann-Hilbert problem. In order to state this, we need to introduce a second family of polynomials $\{q_{j}\}_{j=0}^{\infty}$ satisfying a suitable bi-orthogonality relation with respect to $\{\tilde{p}_{j}\}_{j=0}^{\infty}$. For these purposes it is convenient to deform the contour through $1$ and let $\mathcal{C}$ be the unit circle. As in \cite{DIK11,WW19}, we define these polynomials by the formula
\begin{equation}
\begin{split}
q_{j}(z^{-1}) &= \frac{\prod_{k=0}^{j}\frac{\Gamma(\frac{\gamma}{2}+k+1)\Gamma(\alpha)}{\Gamma(\frac{\gamma}{2}+k+\alpha+1)}}{\sqrt{T_{j}T_{j+1}}}\,\\
&\times \det\begin{pmatrix} \oint_{\mathcal{C}} w(s)\frac{ds}{2is} & \ldots & \oint_{\mathcal{C}}s^{j-1}w(s)\frac{ds}{2is} & 1\\
\vdots & \ldots & \vdots & \vdots\\
\oint_{\mathcal{C}} s^{-j}w(s)\frac{ds}{2is} & \ldots & \oint_{\mathcal{C}}s^{-1}w(s)\frac{ds}{2is} & z^{-j}
\end{pmatrix}
\end{split}
\end{equation}
where $T_{n}=T_n(x,\gamma)$ is the Toeplitz determinant
\begin{equation}\label{eq:TM}
T_{n} = \det\bigg\{\oint_{\mathcal{C}}z^{j-k}w(z)\,\frac{dz}{2i z}\bigg\}_{j,k=0}^{n-1}.
\end{equation}
These polynomials satisfy the orthogonality
\begin{equation}
\oint_{\mathcal{C}}z^{k}q_{j}(z^{-1})w(z)\,\frac{dz}{2iz} = \frac{\delta_{jk}}{\hat\chi_{j}}, \qquad k=0,\ldots,j \label{qorthog}
\end{equation}
and the bi-orthonormality
\begin{equation}
\oint_{\mathcal{C}}\tilde{p}_{k}(z)q_{j}(z^{-1})w(z)\,\frac{dz}{2iz} = \delta_{jk}. \label{biorth}
\end{equation}
The leading coefficient of $q_{j}(z^{-1})$ is $\hat{\chi}_{j}$. By Heine's identity, $T_{n}$ is related to the coefficients and the orthogonal polynomials as follows
\begin{equation}
\begin{split}
T_{n} &= \frac{1}{n!}\oint_{\mathcal{C}^{n}}|\Delta(\bm{z})|^{2}\prod_{j=1}^{n}\frac{w(z_j)\,dz_{j}}{2iz_{j}}\\
&= \frac{1}{n!}\oint_{\mathcal{C}^{n}}\det\bigg\{\frac{\tilde{p}_{j-1}(z_i)}{\chi_{j-1}}\bigg\}_{i,j=1}^{n}\det\bigg\{\frac{q_{j-1}(z_i^{-1})}{\hat{\chi}_{j-1}}\bigg\}_{i,j=1}^{n}\prod_{j=1}^{n}\frac{w(z_j)\,dz_{j}}{2iz_{j}}\\
&=\left(\prod_{j=0}^{n-1}\chi_{j}^{-1}\hat{\chi}_{j}^{-1}\right)\det\bigg\{\oint_{\mathcal{C}}\tilde{p}_{j-1}(z)q_{k-1}(z^{-1})w(z)\frac{dz}{2i z}\bigg\}_{j,k=1}^{n}\\
&=\prod_{j=0}^{n-1}\frac{\Gamma(\frac{\gamma}{2}+j+\alpha+1)}{\Gamma(\frac{\gamma}{2}+j+1)\Gamma(\alpha)}\chi_{j}^{-2}. \label{tnrelation}
\end{split}
\end{equation}
In the above, we used that $\frac{\tilde{p}_{j-1}(z)}{\chi_{j-1}}$ and $\frac{q_{j-1}(z)}{\hat{\chi}_{j-1}}$ are both monic families of polynomials to rewrite the Vandermonde factors. Then we used \eqref{biorth} and the relation between $\chi_{j}$ and $\hat{\chi}_{j}$ in \eqref{norm-rel}. By \eqref{momsgin} and \eqref{tnrelation} we get the relation
\begin{equation}\label{eq:EgammaTM}
\mathbb{E}\left(|\det(B_{n}-x)|^{\gamma}\right) = \frac{n!}{Z_{N,n}}\,\prod_{j=0}^{n-1}\frac{\Gamma(\frac{\gamma}{2}+j+1)\Gamma(\alpha)}{\Gamma(\frac{\gamma}{2}+j+\alpha+1)}T_{n}(x,\gamma).
\end{equation}
We can now state the Riemann-Hilbert problem characterising the above polynomials. Define $\sigma_3$ as the Pauli matrix $\operatorname{diag}(1,-1)$. Define the weight
\begin{equation}\label{eq:wk}
w_{n}(z) = w(z) \frac1{z^n} = (1-x^{2}z)^{\alpha+\frac{\gamma}{2}}\left(\frac{z-1}{z}\right)^{\frac{\gamma}{2}}\,\frac{1}{z^{n}},
\end{equation}
where $\re\gamma > -2$ and fix $0 < x < 1$. Let $Y : \mathbb{C}^{2 \times 2}\setminus\mathcal C \to \mathbb{C}$ be a complex $2 \times 2$ matrix satisfying the properties
\begin{itemize}
\item $Y$ is analytic on $\mathbb{C}\setminus \mathcal{C}$.
\item On $\mathcal{C}$ we have the jump
\begin{equation}
Y_{+}(z) = Y_{-}(z)\begin{pmatrix} 1 & w_{n}(z)\\ 0 & 1 \end{pmatrix}.
\end{equation}
\item As $z \to \infty$ we have the asymptotics
\begin{equation}
Y(z) = \left(I+\mathcal O\left(\frac{1}{z}\right)\right)z^{n\sigma_{3}}, \qquad z \to \infty.
\end{equation}
\item $Y(z)$ is bounded near $z=1$ when $\mathcal C$ does not pass through $1$.\\
In the case that $\mathcal C$ passes through $1$ we impose
\begin{align*}
Y(z) &= \begin{pmatrix}
\mathcal O(1) & \mathcal O((z-1)^\frac{\gamma}{2})\\
\mathcal O(1) & \mathcal O((z-1)^\frac{\gamma}{2})
\end{pmatrix}
& z\to 1.
\end{align*}
\end{itemize}
Then the unique solution to this Riemann-Hilbert problem is given by
\begin{equation}
Y_{n}(z) = \begin{pmatrix} \frac{1}{\chi_{n}}\tilde{p}_{n}(z) & \frac{1}{\chi_{n}}\displaystyle\int_{\mathcal C}\frac{\tilde{p}_{n}(s)}{s-z}w_{n}(s)\frac{ds}{2\pi i}\\
-\chi_{n-1}z^{n-1}q_{n-1}(z^{-1}) & -\chi_{n-1}\displaystyle\int_{\mathcal C}\frac{s^{n}q_{n-1}(s^{-1})}{s-z}w_{n}(s)\frac{ds}{2\pi i u} \end{pmatrix}.
\end{equation}
\begin{remark}
Note that, while the first column of $Y$ is independent of the choice of jump contour $\mathcal C$, the second column is not. This necessitates the extra condition near $z=1$. 
\end{remark}
We also show that the partition function in \eqref{momsgin} can be expressed in terms of the matrix entries of $Y$. Define
\begin{equation} \label{eq:Rgamma}
R_{\gamma}(x) = \mathbb{E}\left(|\det(B_{n}-x)|^{\gamma}\right).
\end{equation}
\begin{lemma}
\label{lem:diffid}
The following differential identity holds:
\begin{equation}
\begin{split}
\frac{d}{dx}(\log R_{\gamma}(x)) &= -2x\left(\frac{\gamma}{2}+\alpha\right)\,\left( -n u^{n+1}q_{n}(u^{-1})I_{12}(u)\right.\\
&\left.-n u+u^{3}[\partial_{z}\tilde{p}_{n}(z)]_{z=u}\,I_{22}(u)\right.\\
&\left.-u^{n+2}[\partial_{z}q_{n}(z^{-1})]_{z=u}\,I_{12}(u)\right) \label{diffid}
\end{split}
\end{equation}
where $u=x^{-2}$ and
\begin{equation}\label{eq:I12I22}
\begin{split}
I_{12}(u) &= \oint_{\mathcal C}\frac{z^{-n+1}\tilde{p}_{n}(z)}{u-z}\frac{w(z)}{2iz}dz\\
I_{22}(u) &= \oint_{\mathcal C}\frac{q_{n}(z^{-1})}{u-z}\frac{w(z)}{2iz}dz.
\end{split}
\end{equation}
\end{lemma}
\begin{proof}
See Appendix \ref{app:diffid}.
\end{proof}
We mention that a similar formula, expressing the partition function in terms of the solution of an associated RHP, was recently obtained in the case of the Ginibre ensemble with point insertion \cite{byun2024free}. 

\section{Steepest descent analysis of the RHP in the strong regime}

\subsection{A suitable $g$ function and first transformation}

Let $n$ and $N=N_n$ be positive integers such that $n<N$. In this section we consider the strong regime, where $n\sim \tilde\mu N$ for some $\tilde\mu\in(0,1)$. We define 
\begin{align*}
c &:= \frac{\alpha}{n}=\frac{N-n}{n} \to \tilde\mu^{-1}-1, & n\to\infty.
\end{align*}
Defining a function (the potential)
\begin{align*}
V(z) &= \log(z) - c \log(1-x^2z), 
\end{align*}
where we give $z\mapsto \log z$ cut $(-\infty,0]$ and $z\mapsto \log(1-x^2z)$ cut $[x^{-2},\infty)$ (being positive for large positive $z$ and large negative $z$ respectively),
we may rewrite our weight \eqref{eq:wk} as
\begin{align*}
w_n(z) &= h_\gamma(z) (1-x^2z)^\frac{\gamma}{2} e^{-n V(z)},
\end{align*}
where $h_\gamma$ is the analytic function with cut $[0,1]$ defined as
\begin{align*}
h_\gamma(z) = \left(\frac{z-1}{z}\right)^{\frac{\gamma}{2}} = \exp\left(\frac\gamma 2(\log(1-1/z))\right).
\end{align*}
Note that $e^{-n V}$ has cut $[x^{-2},\infty)$ (the choice of cut $[0,\infty)$ for $\log z$ will serve a purpose later on though), while the weight $w_n$ has cut $[0,1]\cup[x^{-2},\infty)$. 
%\begin{remark}
%The reader familiar with RHPs may wonder why we define $c=(N-n+\frac\gamma2)/n$ rather than $c=(N-n)/n$ and absorb the extra factor $(1-x^2z)^{\frac\gamma2}$ into the definition of $h_\gamma(z)$. This choice is merely for notational convenience, avoiding larger formulae later on in the analysis. In the RH analysis below we shall assume that $\gamma$, and consequently $c$, is real. The reader may verify that the case of complex $\gamma$ leads to the same asymptotic behaviors, e.g., by absorbing the factor $(1-x^2z)^{\frac\gamma2}$ into $h_\gamma(z)$ (leading only to a small modification of the global parametrix). 
%\end{remark}
Both the weight $w_n$ and $V$ are well-defined on the curve $\mathcal C$, which we shall deform. Our first task is to find a $g$-function to normalize our RHP.
We use $\Gamma_r$ to denote a deformation of the original contour of integration $\mathcal{C}$, depending on a parameter $r\in[1,x^{-2})$ (that we keep free for now) and such that it encircles the interval $[0,1)$ and crosses the real axis at this point $r$. Eventually, we will fix the definition of $\Gamma_r$ (given $r$). 
In order to be useful for the normalization of the RHP, it is well-known that $g$ should satisfy three desired properties. The $g$ function needs to satisfy the asymptotic condition
\begin{align*}
g(z) &= \log z + \mathcal O(1/z), & z\to\infty,
\end{align*}
and furthermore, it needs to satisfy the boundary condition
\begin{equation}
g_{+}(z)+g_{-}(z) = \ell + V(z), \qquad z \in \Gamma_{r},
\end{equation}
for some constant $\ell$. Lastly, we need $e^{\pm n (g_+(z)-g_-(z))}$ to be oscillatory on $\Gamma_r$ as $n\to\infty$. We can explicitly construct a function $g$ that satisfies the first two requirements as
\begin{equation}\label{eq:defg}
g(z) = \begin{cases} -c\log(1-x^{2}z) + \ell & z \in \mathrm{Int}(\Gamma_r)\\
\log(z) & z \in \mathrm{Ext}(\Gamma_r)\setminus (-\infty,0].
\end{cases}
\end{equation}
As before, $z\mapsto \log z$ is taken to have a cut on the negative real line (although the specific choice is irrelevant). We have some freedom left in both $\ell$ and the contour $\Gamma_r$, which we will use to obtain the oscillatory behavior. We now simply define $\Gamma_r$ to be the curve passing through $r$ and encircling $[0,1]$ such that $\re(g_+(z)-g_-(z))=0$ on $\Gamma_r$, which means that
\begin{align} \label{eq:constRealGammar}
|1-x^2z|^c|z| = e^{\re \ell}. 
\end{align}
If we take 
\begin{align} \label{eq:defell}
\ell = \log r + c \log(1-x^2r),
\end{align}
then $\Gamma_r$ wil indeed pass through $r$. 
Now, with $\ell$ as in \eqref{eq:defell}, we define the function $\phi_r$ and $\phi$ via
\begin{equation}\label{eq:phir}
\phi_{r}(z)=\phi(z)-\ell = \log z+c\log(1-x^2 z)-\ell,
\end{equation}
with branch cut $(-\infty,0]\cup[x^{-2},\infty)$.
Then we may alternatively write the condition \eqref{eq:constRealGammar} as $\re \phi_r(z)=0$ or equivalently as the level set
\begin{align*}
\re \phi(z) = \phi(r).
\end{align*}
It is easy to verify that $\phi$ (and thus $\phi_r$) has a stationary point at
\begin{equation}\label{eq:z0}
z_0=\frac{1}{x^2(1+c)}.
\end{equation}
We prove some other useful properties of $\phi_r$ in Appendix \ref{Apx:phiAnalyticIso}.
 Let us consider
\begin{align} \label{eq:levelCurveTildeGamar}
\tilde\Gamma_r = \left\{z\in\mathbb C : |z|=r\left|\frac{1-x^2r}{1-x^2z}\right|^c\right\}
= \left\{z\in\mathbb C : \re\phi(z) = \phi(r)\right\}.
\end{align}
%When $c=0$, $\tilde\Gamma_r$ is the circle $|z|=r$. 
(Here $\re \phi$ is extended to $(-\infty,0)\cup (x^{-2},\infty)$ by continuity.) Our interest is in the case $c>0$, and then we have the following lemma. It turns out that $\tilde\Gamma_r$ consists of two Jordan curves, one of which will be used for the definition of $\Gamma_r$. We state the following result for $0<r\leq z_0$, although we shall only need it for $1\leq r\leq z_0$ in our RH analysis. 

\begin{lemma} \label{eq:GammarJordan}
Let $c>0$ and $0<r\leq z_0$. Then $\tilde\Gamma_r$ as defined in \eqref{eq:levelCurveTildeGamar} consists of two closed Jordan curves, which both intersect the real line in exactly two points. They can only intersect eachother in the case $r=z_0$, in which case $z_0$ is their only point of intersection. One of these Jordan curves, $\Gamma_r$, encloses the interval $[0,r)$ and passes through $r$.\\
Furthermore, $\Gamma_{r'}\subset \operatorname{Int }\Gamma_{r}$ when $0<r'<r$.
\end{lemma}

\begin{figure}
\centerline{\includegraphics[scale=1]{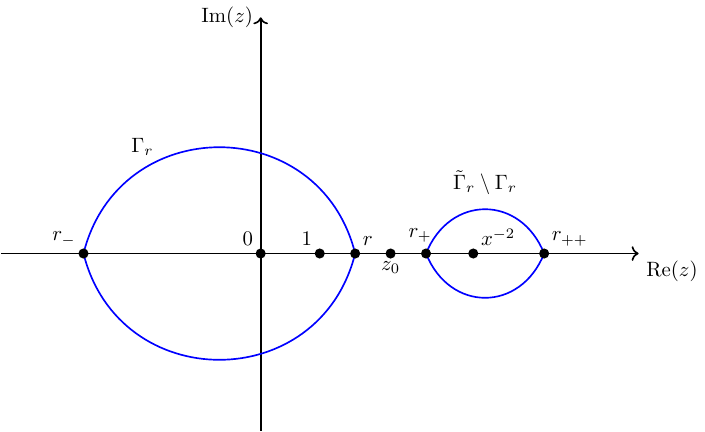}}
\caption{Graphic representation of the level curve $\tilde \Gamma_r$ as defined in \eqref{eq:levelCurveTildeGamar}, here for $r\neq z_0=\frac1{(c+1)x^2}$.}
\label{fig:levelCurveGammar}
\end{figure}

\begin{proof}
The case $c=1$ can be solved explicitly, and the case $c>1$ can be transformed $c\to 1/c$. Thus it suffices to prove our result for $c\in(0,1)$. The map $t\mapsto |\phi(t)|= |t| |1-x^2 t|^{c}$ is a bijection both from $(0,z_0)$ and $(z_0,x^{-2})$ onto $(0,\phi(z_0))$, and a bijection both from $(-\infty,0)$ and $(x^{-2},\infty)$ onto $(0,\infty)$. Therefore, given $r\in(0, z_0)$, there are precisely four real points $r_-, r, r_+$ and $r_{++}$ on $\tilde\Gamma_r$. We have $r_-<0$ and $r_{++}>r_+>r$ (see Figure \ref{fig:levelCurveGammar}). In the case $r=z_0$, we have $r=r_+$ and there are only three real points on $\tilde\Gamma_r$. It is a direct consequence of Lemma \ref{lem:phiAnalyticIso} that there exist sets $\mathcal S_1, \mathcal S_2, \mathcal S_3$ and the complex conjugates $\mathcal S_2^*, \mathcal S_3^*$ such that $\phi$ is a bijection from $\mathcal S_1\setminus (-\infty,0]$ to $|\operatorname{Im} \zeta| \in (\pi c,\pi)$, a bijection from both $\mathcal S_2$ and $\mathcal S_3^*$ to $\operatorname{Im} \zeta\in (0,\pi c)$, and a bijection from both $\mathcal S_2^*$ and $\mathcal S_3$ to $\operatorname{Im} \zeta\in (-\pi c,\pi)$, see Figure \ref{fig:levelCurveGammar2} for a graphic representation. Since $\phi$ is analytic with non-zero derivative in every point of $\partial S_1\setminus\{0\}$ we can patch regions together and conclude that $\phi$ is an analytic isomorphism from $(\overline{\mathcal S_1}\cup \mathcal S_2 \cup \mathcal S_2^*)\setminus (-\infty,0]$ to $|\operatorname{Im} \zeta|<\pi$. Similarly, $\phi$ is an analytic isomorphism from $\mathcal S_3\cup \mathcal S_3^*\cup (z_0,x^{-2})$ to $|\operatorname{Im} \zeta|<\pi c$.

Now consider the vertical line segment $L_r$ consisting of $\zeta$ with $\re\zeta=\phi(r)$ and $\operatorname{Im} \zeta\in(-\pi,\pi)$. By the preceding, $\phi^{-1}(L_r)\cap ((\overline{\mathcal S_1}\cup \mathcal S_2 \cup \mathcal S_2^*)\setminus (-\infty,0])$ is a smooth non-self-intersecting curve which starts and ends in $r_-$, and it necessarily contains $r$. On the other hand, $\phi^{-1}(L_r)\cap (\mathcal S_3\cup \mathcal S_3^*\cup (z_0,x^{-2}))$ is a smooth non-self-intersecting curve in $\mathcal S_3\cup \mathcal S_3^*\cup (z_0,x^{-2})$, which necessarily intersects the real line in $r_+$ and begins and ends in $r_{++}$. Adding $r_-$ and $r_{++}$ to these curves makes them closed and thus a Jordan curve. The last part of the lemma follows from the fact that $\phi$ is an analytic isomorphism from $(\overline{\mathcal S_1}\cup \mathcal S_2 \cup \mathcal S_2^*)\setminus (-\infty,0]$ to $|\operatorname{Im} \zeta|<\pi$.
\end{proof}

\begin{figure}
\centerline{\includegraphics[scale=1]{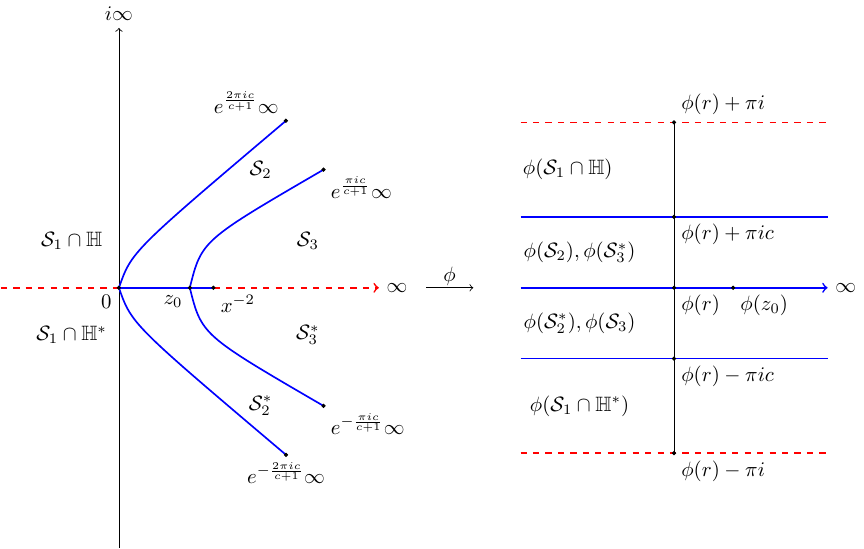}}
\caption{The image under $\zeta=\phi(z)$ of the sets $\mathcal S_1\setminus (-\infty,0), \mathcal S_2, \mathcal S_2^*, \mathcal S_3$ and $\mathcal S_3^*$ from Lemma \ref{lem:phiAnalyticIso}, which consists of horizontal strips. The inverse image of the vertical line segment with $\re\zeta=\phi(r)$ is given by the two Jordan curves in Figure \ref{fig:levelCurveGammar} (not shown in the current figure). The interval $(-\infty,0]$ corresponds to the lines $\operatorname{Im} \zeta=\pm \pi i$, while the interval $[x^{-2},\infty)$ corresponds to $\operatorname{Im} \zeta=\pm \pi i c$ (but so does the curve $\partial S_1$). The intervals $[0,z_0]$ and $[z_0,x^{-2}]$ correspond to $\zeta\in (-\infty,\phi(z_0)]$. The curve formed by $\partial S_2\cap \partial S_3$ and its complex conjugate corresponds to $\zeta\in [\phi(z_0),\infty)$.}
\label{fig:levelCurveGammar2}
\end{figure}

Finally then, given $r\in(0,z_0]$, the definition of the $g$-function is fixed, and given by \eqref{eq:defg} with the choice $\Gamma_r$ from Lemma \ref{eq:GammarJordan}. 

As the reader may verify using the residue theorem, the $g$-function can alternatively be expressed as
\begin{equation}\label{eq:defg}
g(z) = \oint_{\Gamma_r}\log(z-w)d\nu(w)
\end{equation}
where we give $\Gamma_r$ positive orientation, and $d\nu$ is the probability measure defined through
\begin{align*}
d\nu(z)
=
\frac{1}{2\pi i}\left(\frac{1}{z}-\frac{cx^2}{1-x^2 z}\right) dz, \qquad z\in \Gamma_r.
\end{align*}
Here the logarithm is defined with a branch cut on $[0,\infty)$.
%Similarly to \cite[Section 3.1.1]{BGM15}, we can calculate the constant $\ell$ via direct integration of \eqref{eq:defg} on a circle of radius $r>0$ and with a fixed value of $z$, for example $z=0$. This gives 
%\begin{equation}\label{eq:ell}
%g(0)
%=
%\ell
%=
%\log r +c\log(1-x^2r).
%-c\log(1+x^2r)
%=
%\log r +c\log\left(\frac{1-x^2r}{1+x^2r}\right).
%\end{equation}
A useful observation is that $\phi_r$ satisfies the following relation.
\begin{equation}
\phi_r(z)
=
\begin{cases}
-2g(z)+V(z)+\ell,&\qquad z\in\textrm{Int}(\Gamma_r)\setminus[0,r),\\
2g(z)-V(z)-\ell,&\qquad z\in\textrm{Ext}(\Gamma_r)\setminus(r,\infty),
\end{cases}
\end{equation}
Furthermore, for $z\in\Gamma_r$, we have
\begin{equation}
\phi_r(z)=-g_+(z)+g_-(z).
\end{equation}
%We can also check that
%\begin{equation}
%d\nu(z)
%=
%\frac{1}{2\pi i}\frac{d\phi(z)}{dz}.
%\end{equation}       
%
A local expansion around the stationary point $z_0$ gives 
\begin{equation}\label{eq:phiatz0}
\phi_r(z)=\phi_r(z_0)-\frac{(1+c)^3 x^4}{c}(z-z_0)^2+\mathcal{O}((z-z_0)^3)
\end{equation}
as $z\to z_0$. It follows that (locally) we have two directions of steepest ascent with angle $\pm\frac{\pi}{2}$ and two directions of steepest descent with angle $0,\pi$ from the stationary point. The fact that $z_0<x^{-2}$ follows directly from the fact that $c>0$. If we want $z_0>1$, then we need to impose that $x^2(1+c)<1$. In the strong regime, we have $\frac{n}{N}\to\tilde\mu$ as $n\to\infty$, so $c\to\tilde\mu^{-1}-1$ as $n\to\infty$, therefore the condition holds for $0<x<\sqrt{\tilde\mu}$, that is, for $x$ in the bulk of the spectrum. As $x\uparrow\sqrt{\tilde\mu}$, we have $z_0\downarrow 1$.

%\textcolor{magenta}{Now we need an analysis of the family of contours $\Gamma_r$, depending on $r$, to find the right deformation of the original $\Sigma$. This should be compatible with the branch cuts that we had.}
%
%\textcolor{magenta}{If I understand it correctly, we fix $c$ (which depends on $N$, $M$, $\gamma$ and $k$) and $x$, (which satisfies $0<x<1$), and then we should be able to determine a suitable $r$.}

%It is important to notice that we take $\log z$ with a cut on the positive real axis (I think that this is done in order to avoid extra jumps on the negative real axis when making transformations). This means that
%\begin{equation}
%\log z=\log |z|+i\theta, \qquad \theta\in(0,2\pi),
%\end{equation}
%and (with $\mathbb{R}^+$ oriented from the origin towards infinity), we have $\phi_{r+}(x)=\phi_{r-}(x)-2\pi i$.

Now that we have defined the $g$ function, we can define the normalization transformation. For notational convenience, it will turn out to be nice to modify the $g$-function and $\phi$, we define
\begin{align*}
\tilde g(z) &= \begin{cases} -(c+\frac\gamma{2n})\log(1-x^{2}z) + \tilde\ell & z \in \mathrm{Int}(\Gamma_r)\\
\log(z) & z \in \mathrm{Ext}(\Gamma_r)\setminus (-\infty,0],
\end{cases}\\
\tilde\phi_r(z) &= \tilde\phi(z) - \tilde\ell = \log z + (c+\frac\gamma{2n}) \log(1-x^2z) - \tilde\ell,\\
\tilde\ell &= \ell+\frac{\gamma}{2n}\log(1-x^2r).
\end{align*}
(In other words, $c$ is replaced by $c+\frac\gamma{2n}$.) We define $T$ directly from $Y$:
%If we do not introduce $\tilde{Y}$, but instead define $T$ directly from $Y$, we have
\begin{equation}
T(z)=e^{-\frac{n\tilde\ell}{2}\sigma_3}Y(z)e^{-n\left(\tilde g(z)-\frac{\tilde\ell}{2}\right)\sigma_3}.
\end{equation}
Notice that $T$ is well-defined on $[r,\infty)$, since the cut of $\log z$ disappears in the composition $e^{\pm n \tilde g(z)}$. In what follows we give $\Gamma_r$ positive orientation, which means that the $+$ side is in its interior while the $-$ sign is in its exterior. 
$T(z)$  satisfies the following Riemann--Hilbert problem, where $r\in[1,z_0]$:
\begin{itemize}
\item $T$ is analytic on $\mathbb{C}\setminus \Gamma_r$.
\item On $\Gamma_r$, the matrix $T(z)$ has jumps
\begin{equation}
T_{+}(z) = T_{-}(z)
\begin{pmatrix} e^{n\tilde\phi_r(z)} & h_\gamma(z)\\ 0 & e^{-n\tilde\phi_r(z)} \end{pmatrix}, \,\, z\in\Gamma_r.
\end{equation}
\item As $z \to \infty$ we have the asymptotics
\begin{equation}
T(z) = I+\mathcal O\left(\frac{1}{z}\right).
\end{equation}
\item $T(z)$ is bounded near $z=1$ when $\Gamma_r$ does not pass through $1$.\\
In the case that $r=1$ we impose
\begin{align*}
T(z) &= \begin{pmatrix}
\mathcal O(1) & \mathcal O((z-1)^\frac{\gamma}{2})\\
\mathcal O(1) & \mathcal O((z-1)^\frac{\gamma}{2})
\end{pmatrix}
& z\to 1.
\end{align*}
\end{itemize}
\subsection{Opening of lenses}
Here we assume that $1 \leq r_{1} \leq r \leq r_{2}<x^{-2}$, where $r\in[1,z_0]$, and we open lenses $\Sigma_{r_1}$ and $\Sigma_{r_2}$ in the interior and exterior of $\Gamma_r$ respectively. We pick $\Sigma_{r_2}$ in such a way that that it does not intersect with the interior of the component $\tilde\Gamma_r\setminus \Gamma_r$ (see Lemma \ref{eq:GammarJordan}). $\Sigma_{r_1}$ is a positively oriented Jordan curve passing through $r_1$ which lies in the interior of $\Gamma_r$, except possibly for the point $r_1$. $\Sigma_{r_2}$ is a postively oriented Jordan curve passing through $r_2$ which lies in the exterior of $\Gamma_r$, except possibly for the point $r_2$. See Figure \ref{fig:GammaS_Ken} below. 
The jump matrix can be factorized:
\begin{eqnarray*}
\pmtwo{e^{n\tilde\phi_r(z)}}{h_\gamma(z)}{0}{e^{-n\tilde\phi_r(z)}} =\hspace{1.5in}
\ \ \ \ \ \  && \\
\pmtwo{1}{0}
{e^{-n\tilde\phi_r(z)}h_{-\gamma}(z)}{1} 
\pmtwo{0}{h_\gamma(z)}
{-h_{-\gamma}(z)}{0}
\pmtwo{1}{0}{e^{n\tilde\phi_r(z)}h_{-\gamma}(z)}{1}.
\end{eqnarray*}
Now define $\Omega_1$ as the open region bounded by $\Sigma_{r_1}$ and $\Gamma_r$, $\Omega_2$ as the open region bounded by $\Gamma_r$ and $\Sigma_{r_2}$, and by $\Sigma_\infty$ we define the remaining region, as in Figure \ref{fig:GammaS_Ken}. Now we define
%instead of the definition \eqref{eq:TtoS}, we define
\begin{equation}\label{eq:6.7}
S(z)
=
T(z)
\begin{cases}
I, & \qquad z\in\Omega_{\infty},\\
\pmtwo{1}{0}{e^{n\tilde\phi_r(z)}h_{-\gamma}(z)}{1}^{-1}, & \qquad z\in\Omega_1,\\
\pmtwo{1}{0}
{e^{-n\tilde\phi_r(z)} h_{-\gamma}(z)}{1} , & \qquad z\in\Omega_2.
\end{cases}
\end{equation}
The quantity $S$ satisfies the Riemann-Hilbert problem:
\begin{itemize}
\item $S$ is analytic on $\mathbb{C}\setminus \Sigma_S$, where $\Sigma_S=\Sigma_{1}\cup\Gamma_{r}\cup\Sigma_{2}$, as shown in Figure \ref{fig:GammaS_Ken} (note that there is no jump across $(0,1)$).
\begin{figure}
\centerline{\includegraphics[scale=1]{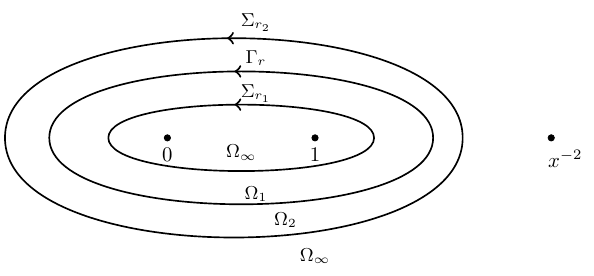}}
\caption{Contour $\Gamma_S$, for the case of strict inequalities $1<r_1<r<r_2<x^{-2}$.}
\label{fig:GammaS_Ken}
\end{figure}
\item On $\Sigma_S$, the matrix $S(z)$ has jumps
\begin{equation}
S_{+}(z) = S_{-}(z)
\begin{cases} 
\pmtwo{1}{0}
{e^{n\tilde\phi_r(z)}h_{-\gamma}(z)}{1} , & \,\, z\in\Sigma_{1},\\%\cup(x^{-2},\infty),\\
%\begin{pmatrix} e^{-k(g_+(z)-g_-(z))} & 1\\ 0 & e^{k(g_+(z)-g_-(z))} \end{pmatrix}, & \,\, z\in\Gamma_r.
%\begin{pmatrix} 1 & 0\\ e^{n\phi(z)} & 1 \end{pmatrix}, & \,\, z\in\Gamma_{r_1},\\
\pmtwo{0}{h_\gamma(z)}
{-h_{-\gamma}(z)}{0}, & \,\, z\in\Gamma_{r},\\
\pmtwo{1}{0}{e^{-n\tilde\phi_r(z)}h_{-\gamma}(z)}{1}, & \,\, z\in\Sigma_{2}.
\end{cases}
\end{equation}
\item As $z \to \infty$ we have the asymptotics
\begin{equation}
S(z) = I+\mathcal O\left(\frac{1}{z}\right).
\end{equation}
\item $S(z)$ is bounded near $z=1$ when $r_1>1$ (and necessarily $r>1$).\\
In the case that $r_1=1$ we impose
\begin{align*}
S(z) &= \begin{pmatrix}
\mathcal O(1) & \mathcal O((z-1)^\frac{\gamma}{2})\\
\mathcal O(1) & \mathcal O((z-1)^\frac{\gamma}{2})
\end{pmatrix}
& z\to 1.
\end{align*}
\end{itemize}
Note that $\re\phi_r(z)=0$ is equivalent to $z\in\Gamma_r$, due to Lemma \ref{eq:GammarJordan} and its proof. Furthermore, in the case of strict inequalities $r_1<r<r_2$ we may pick $\Sigma_{r_1}=\Gamma_{r_1}$ and $\Sigma_{r_2}=\Gamma_{r_2}$, and Lemma \ref{eq:GammarJordan} implies that $\re \phi_r<0$ on $\Sigma_{r_1}$, and $\re\phi_r>0$ on $\Sigma_{r_2}$ when $r_2> r$ is close enough to $r$ (see the proof of Lemma \ref{eq:GammarJordan}, $r_2\in (r,r_+)$). Therefore, the jump matrices corresponding to $\Gamma_{r_1}$ and $\Gamma_{r_2}$ tend to the unit matrix as $n\to\infty$, in the case that $r_1<r<r_2$ and $r_2>r$ is close enough to $r$ (note that $e^{\pm n \phi(z)}$ and $e^{\pm n \tilde\phi(z)}$ differ only be a factor of order $1$). In the case that these equalities are not strict, a local parametrix problem may have to be considered, although we shall present an alternative approach below. 

\subsection{Global parametrix}
The global parametrix is the solution of the Riemann-Hilbert problem for $P^{(\infty)}(z)$, where the jumps of $S$ on $\Sigma_{r_1}$ and $\Sigma_{r_2}$ are ignored:
\begin{itemize}
\item $P^{(\infty)}$ is analytic on $\mathbb{C}\setminus \Gamma_{r}$.%\cup\Gamma_{r}$.
\item On $\Gamma_{r}$, the matrix $P^{(\infty)}(z)$ has jump
\begin{equation}
P^{(\infty)}_{+}(z) = P^{(\infty)}_{-}(z)\pmtwo{0}{h_\gamma(z)}
{-h_{-\gamma}(z)}{0}.
\end{equation}
\item As $z \to \infty$ we have the asymptotics
\begin{equation}
P^{(\infty)}(z) = I+\mathcal O\left(\frac{1}{z}\right).
\end{equation}
%\item As $z\to 0 $ and as $z\to 1$, $P^{(\infty)}(z)$ has the same local behavior as $\widetilde{S}(z)$.
\end{itemize}
The solution to this Riemann-Hilbert problem is explicit:
\begin{eqnarray*}
P^{(\infty)}(z) =  
\begin{cases} 
\pmtwo{0}{1}{-1}{0}, & \,\, z\in \mbox{Int}(\Gamma_{r}),\\%\cup(x^{-2},\infty),\\
%\begin{pmatrix} e^{-k(g_+(z)-g_-(z))} & 1\\ 0 & e^{k(g_+(z)-g_-(z))} \end{pmatrix}, & \,\, z\in\Gamma_r.
%\begin{pmatrix} 1 & 0\\ e^{n\phi(z)} & 1 \end{pmatrix}, & \,\, z\in\Gamma_{r_1},\\
\pmtwo{h_{-\gamma}(z)}{0}{0}
{h_\gamma(z)}, & \,\, z\in\mbox{Ext}(\Gamma_{r}).
\end{cases}
\end{eqnarray*}

\subsection{Final transformation}

Again, let $1\leq r_1\leq r\leq r_2<x^{-2}$, where $r\in[1,z_0]$. Now let $\Sigma_{R}=\Sigma_{r_1}\cup \Sigma_{r_2}$. We next define $R$:
\begin{eqnarray}
R(z) = S(z) \left( P^{(\infty)}(z) \right)^{-1},
\end{eqnarray}
so that $R(z)$ satisfies the following Riemann--Hilbert problem:
\begin{itemize}
\item $R$ is analytic on $\mathbb{C}\setminus \Sigma_{R}$.
\item On $\Sigma_{R}$, the matrix $R(z)$ has jumps
\begin{equation}
R_{+}(z) = R_{-}(z)
\begin{cases}
\pmtwo{1}{-e^{n\tilde\phi_r(z)}h_{-\gamma}(z)}{0}{1}, & \,\, z\in\Sigma_{r_1},\\%\cup(x^{-2},\infty),\\
\pmtwo{1}{0}{e^{-n\tilde\phi_r(z)}h_\gamma(z)}{1}, & \,\, z\in\Sigma_{r_2}.
\end{cases}
\end{equation}
\item As $z \to \infty$ we have the asymptotics
\begin{equation}
R(z) = I+\mathcal O\left(\frac{1}{z}\right).
\end{equation}
\item $R(z)$ is bounded near $z=1$ when $r_1>1$.\\
In the case that $r_1=1$ we impose
\begin{align*}
R(z) &= \begin{pmatrix}
\mathcal O((z-1)^\frac{\gamma}{2}) & \mathcal O(1)\\
\mathcal O((z-1)^\frac{\gamma}{2}) & \mathcal O(1)
\end{pmatrix}
& z\to 1.
\end{align*}
\end{itemize}

The case $r_1=1$ presents several difficulties. Namely, the jump matrix is not in $L^2(\Sigma_{r_1})$ when $\gamma>0$, and $R$ is unbounded near $z=1$ when $\gamma\in(-2,0)$. For this reason we exclude $r_1=1$ in what follows (but it will be important in the double scaling regime, which we postpone to a future work). 
In the case that $1< r_1<r<r_2<x^{-2}$, $R$ is in the small norm setting, provided that $r_2$ is close enough to $r$ (see the proof of Lemma \ref{eq:GammarJordan}). Without loss of generality, we take $\Sigma_1=\Gamma_{r_1}$ and $\Sigma_2=\Gamma_{r_2}$. The fact that $R$ is in the small norm setting can be used to show that
\begin{multline} \label{eq:Restimate}
R(z) = \begin{pmatrix} 1 & - \frac1{2\pi i}\oint_{\Gamma_{r_1}} e^{n\tilde\phi_r(s)} h_{-\gamma}(s) \frac{ds}{s-z}\\ 
\frac1{2\pi i}\oint_{\Gamma_{r_2}} e^{-n\tilde\phi_r(s)} h_\gamma(s) \frac{ds}{s-z} & 1 \end{pmatrix}\\
+e^{-n(\phi_r(r_2)-\phi_r(r_1))} \mathcal O\begin{pmatrix} 1 & e^{n\phi_r(r_1)}\\ e^{-n\phi_r(r_2)} & 1\end{pmatrix}
\end{multline}
uniformly on $\mathbb C\setminus\Sigma_{R}$ as $n\to\infty$, we elaborate on this in the subsequent section. Here the $\mathcal O$-term is to be read entrywise, and depends implicitly on $x$ and $\gamma$. The difference $\phi_r(r_2)-\phi_r(r_1)$ is independent of $r$ and positive when $r_2>r$ is close enough to $r$. The $\mathcal O$-term is thus exponentially small in this setting. 

\subsection{Small norm theory} \label{sec:smallNorm}
We now consider the setting $1<r_1<r\leq r_2\leq z_0$. In principle, we can consider $r_2>z_0$, but the aforementioned cases will suffice for our purposes. We get the sharpest bounds for $e^{n(\phi_r(r_1)-\phi_r(r_2))}$ when $r_1\downarrow 1$ and $r_2= z_0$. Let us derive the behavior of $R$ explicitly, under the assumption $1<r_1<r<r_2\leq z_0$, where $r_2$ is assumed to be close enough to $r$ (thus, excluding $r=r_2$ momentarily). The associated Riemann-Hilbert problem for $R(z)$ is in the small norm setting, since
\begin{align*}
J_R - I= \begin{cases}
\begin{pmatrix}
0 & - e^{n\tilde\phi_r(z)} h_{-\gamma}(z)\\
0 & 0
\end{pmatrix} & z\in \Gamma_{r_1},\\
\begin{pmatrix}
0 & 0\\
e^{-n\tilde\phi_r(z)} h_\gamma(z) & 0
\end{pmatrix} & z\in \Gamma_{r_2}.
\end{cases}
\end{align*}
is locally analytic in a neighborhood of the contour, and satisfies
\begin{eqnarray}
\label{eq:supbound}
\sup_{z\in \Sigma_R} \| J_{R} - I \| \le C e^{n \max(\phi_r(r_{1}) , -\phi_r(r_2))} ,
\end{eqnarray}
for some constant $C>0$ (implicitly depending on $x$ and $\gamma$), where $\| \cdot \|$ means matrix norm.  Note that we have $\phi_r(r_{1}) < 0$ and, since $r_2$ is close enough to $r$, $\phi_r(r_2)>0$, and we conclude that the jump matrix is exponentially close to $I$.  

The Riemann-Hilbert problem for $R$  can be turned into an equivalent system of integral equations by first writing
\begin{eqnarray}
R_{+}-R_{-} = R_{-} \left(J_{R} - I \right) \ , \ z \in \Sigma_R \ ,
\end{eqnarray}
and then using the Cauchy-Plemelj-Sokhotski formula:
\begin{eqnarray}
\label{eq:RRep}
R(z) = I + \frac{1}{2 \pi i} \int_{\Sigma_R}\frac{R_{-}(s) \left( J_{R}(s) - I \right) }{s-z}ds ,
\end{eqnarray}
and taking $-$ boundary values to arrive at
\begin{eqnarray}
R_{-} = I + \left(\frac{1}{2 \pi i} \int_{\Sigma_R}\frac{R_{-}(s) \left( J_{R}(s) - I \right) }{s-z}ds \right)_{-}
\end{eqnarray}
which can be expressed as
\begin{eqnarray}
\left( 1 - C_{J_{R}} \right)(R_{-}) = I,
\end{eqnarray}
\label{eq:inteq}
with the operator $C_{J_{R}}$ defined via
\begin{eqnarray}
&&
C_{J_{R}}(F) = \left(\frac{1}{2 \pi i} \int_{\Sigma_R}\frac{F(s) \left( J_{R}(s) - I \right) }{s-z}ds \right)_{-} \\
&&
=C_{-}\left( \left( J_{R} - I \right) F \right) \ .
\end{eqnarray}
in the last formula above, $C_{-}$ is the usual Cauchy projection operator, 
\begin{eqnarray}
C_{-}(f) = \lim_{\substack{z' \to z \\ z' \in -\mbox{ side of }  \Sigma_R}}\frac{1}{2\pi i}\int_{\Sigma_R} \frac{f(s)}{s-z} ds
\end{eqnarray}
Since $C_{-}$ is bounded on the space $L^{2}(\Sigma_R)$, we learn using \eqref{eq:supbound} that $C_{J_{R}}$ has a norm converging to $0$ as $n\to\infty$: 
\begin{eqnarray}
\vertiii{C_{J_{R}}}_{L^{2}(\Sigma_R)\mapsto L^{2}(\Sigma_R}) \le C e^{n\max(\phi_r(r_{1}),-\phi_r(r_2)) },
\end{eqnarray}
for some constant $C=C_{J_R}$ that depends implicitly on $J_R$ (and thus on $x$ and $\gamma$). 
There is a unique solution $R_{-}$ to \eqref{eq:inteq}, which is given explicitly by a Neumann series
\begin{eqnarray}
R_{-} = \sum_{j=0}^{\infty} \left(C_{J_{R}} \right)^{j}(I) \ , 
\end{eqnarray}
and returning to the representation \eqref{eq:RRep} of $R(z)$, we have
\begin{eqnarray}
R(z) = I + \frac{1}{2 \pi i} \int_{\Sigma_R}\frac{\left(
\sum_{j=0}^{\infty} \left(C_{J_{R}} \right)^{j}(I) (s)
\right) \left( J_{R}(s) - I \right) }{s-z}ds .
\end{eqnarray}
By induction, one may show that
\begin{multline*}
\sum_{j=0}^{\infty} \left(C_{J_{R}} \right)^{j}(I)(s)
= \sum_{j=0}^\infty \begin{pmatrix}
C_{\Sigma_{r_1}} C_{\Sigma_{r_2}})^{j}(1)(s) & 0\\
0 & (C_{\Sigma_{r_2}} C_{\Sigma_{r_1}})^{j}(1)(s)
\end{pmatrix}\\
+ \sum_{j=0}^\infty
\begin{pmatrix}
0 & (C_{\Sigma_{r_1}} C_{\Sigma_{r_2}})^{j}C_{\Sigma_{r_1}}(1)(s)\\
(C_{\Sigma_{r_2}} C_{\Sigma_{r_1}})^{j}C_{\Sigma_{r_2}}(1)(s) & 0
\end{pmatrix},
\end{multline*}
where we define
\begin{align*}
C_{\Sigma_{r_1}}(f)(z) &= -\frac1{2\pi i} \oint_{\Sigma_{r_1}} e^{n\tilde\phi_r(s)} h_{-\gamma}(s) f(s) \frac{ds}{s-z}\\
C_{\Sigma_{r_2}}(f)(z) &= \frac1{2\pi i} \oint_{\Sigma_{r_2}} e^{-n\tilde\phi_r(s)} h_\gamma(s) f(s) \frac{ds}{s-z}.
\end{align*}
Thus we have for $z\in\mathbb C\setminus \Sigma_R$
\begin{multline*}
\frac{1}{2 \pi i} \int_{\Sigma_R}\frac{\left(
\sum_{j=0}^{\infty} \left(C_{J_{R}} \right)^{j}(I) (s)
\right) \left( J_{R}(s) - I \right) }{s-z}ds\\
= \sum_{j=0}^\infty
\begin{pmatrix}
0 & (C_{\Sigma_{r_1}} C_{\Sigma_{r_2}})^{j}C_{\Sigma_{r_1}}(1)(s)\\
(C_{\Sigma_{r_2}} C_{\Sigma_{r_1}})^{j}C_{\Sigma_{r_2}}(1)(s) & 0
\end{pmatrix}\\
+ \sum_{j=1}^\infty \begin{pmatrix}
C_{\Sigma_{r_1}} C_{\Sigma_{r_2}})^{j}(1)(s) & 0\\
0 & (C_{\Sigma_{r_2}} C_{\Sigma_{r_1}})^{j}(1)(s)
\end{pmatrix}
\end{multline*}
and expanding up to order $1$ and bounding the remaining terms, we obtain \eqref{eq:Restimate}, that is
\begin{multline*} 
R(z) = \begin{pmatrix} 1 & - \frac1{2\pi i}\oint_{\Gamma_{r_1}} e^{n\tilde\phi_r(s)} h_{-\gamma}(s) \frac{ds}{s-z}\\ 
\frac1{2\pi i}\oint_{\Gamma_{r_2}} e^{-n\tilde\phi_r(s)} h_\gamma(s) \frac{ds}{s-z} & 1 \end{pmatrix}\\
+e^{-n(\phi_r(r_2)-\phi_r(r_1))} \mathcal O\begin{pmatrix} 1 & e^{n\phi_r(r_1)}\\ e^{-n\phi_r(r_2)} & 1\end{pmatrix},
\end{multline*}
as $n\to\infty$, uniformly on $\mathbb C\setminus \Sigma_R$. 
Let us consider the special case $r=r_2=z_0$ and $r_1>1$ close (enough) to $1$. The corresponding RHP is not in the small norm setting (at least not in an obvious way), since $\phi_r(r_2)=0$ in this case. We note that $\phi_r(r_1)<0$ for $r_1$ close enough to $1$. In particular, we can pick a fixed $\epsilon>0$ such that $R_\epsilon(z) = e^{n\epsilon \sigma_3} R(z) e^{-n\epsilon \sigma_3}$ is in the small norm setting. Then one may repeat the steps above for $R_\epsilon$ and, transforming back $R_\varepsilon\mapsto R$, reach the same conclusion, namely that \eqref{eq:Restimate} also holds in the case $r=r_2=z_0$, if $r_1>1$ is close enough to $1$. Moving back to the general case $1<r_1\leq r\leq r_2\leq z_0$, the question remains whether there is an optimal choice for $r$ (or $\ell$ equivalently). It appears there is no optimal value for $r$, since a simple transformation
\begin{align} \label{eq:defTildeR}
\widetilde R(z) = e^{n\tilde\ell\sigma_3/2} R(z) e^{-n\tilde\ell\sigma_3/2}
\end{align}
removes the dependence on $r$ in \eqref{eq:Restimate}. Namely, we have uniformly for $z$ in compact sets away from $\Sigma_R$
\begin{multline} \label{eq:tildeRestimate}
\widetilde R(z) = \begin{pmatrix} 1 & - \frac1{2\pi i}\oint_{\Gamma_{r_1}} e^{n\tilde\phi(s)} h_{-\gamma}(s) \frac{ds}{s-z}\\ 
\frac1{2\pi i}\oint_{\Gamma_{r_2}} e^{-n\tilde\phi(s)} h_\gamma(s) \frac{ds}{s-z} & 1 \end{pmatrix}\\
+e^{-n(\phi(r_2)-\phi(r_1))} \mathcal O\begin{pmatrix} 1 & e^{n\phi(r_1)}\\ e^{-n\phi(r_2)} & 1\end{pmatrix},
\end{multline}
where we remind the reader that $\phi$ is the $r$-independent function
\begin{equation} \label{eq:defPhiIndepr}
\phi(z) = \phi_r(z)+\ell = \log z + c \log(1-x^2z),
\end{equation}
and it has a cut at $(-\infty,0]\cup [x^{-2},\infty)$ (although notice that $e^{\pm n\phi(z)}$ only has a cut at $[x^{-2},\infty)$), and $\tilde\phi(z)=\phi(z)+\frac{\gamma}{2n}\log(1-x^2z)$.
The best estimates correspond to $r_1\downarrow 1$ and $r_2=z_0$. Of course, while the choice of $r$ is irrelevant for the estimates of $\widetilde R$, it does determine the different regions in which we can look at the large $n$ behavior of $Y$ (due to the normalization transformation $Y\mapsto T$).

\subsection{Undoing transformations}
We are in the setting $1<r_1\leq r\leq r_2\leq z_0$ where $r_1$ is close enough to $1$ and $r_2$ is close enough to $r$. Inverting the transformation of our RH analysis, we deduce the following relation:
\begin{align} \label{eq:asympBehavY}
Y(z) = \widetilde
R(z) e^{n\tilde\ell\sigma_3/2} P^{(\infty)}(z) e^{-n\tilde\ell\sigma_3/2}
\left.
\begin{cases}
\pmtwo{1}{0}
{-e^{-n\tilde\phi(z)}h_{-\gamma}(z)}{1} , & z \in \Omega_{1}, \\
\pmtwo{1}{0}{e^{n\tilde\phi(z)}h_{-\gamma}(z)}{1},  & z \in \Omega_{2}, \\
I , \ & z \in \Omega_{\infty} .
\end{cases} \right\}e^{ n\tilde g(z) \sigma_{3} } \ ,
\end{align}
where $\widetilde R$ is defined in \eqref{eq:defTildeR}, and $\tilde\phi(z)=\log z+(c+\frac\gamma{2n}) \log(1-x^2z)$. 
Our three contours $\Sigma_{r_1}, \Gamma_{r}, \Sigma_{r_2}$ divide $\mathbb C$ into four regions. We will describe the behavior of $Y$ in these regions. First, let us assume that $z\in \operatorname{Int }\Sigma_{r_1}$ where $\Sigma_{r_1}=\Gamma_{r_1}$. First computing $\widetilde R^{-1} Y$ via \eqref{eq:asympBehavY}, then inserting \eqref{eq:tildeRestimate}, and some rearranging yields
\begin{multline} \label{eq:YinIntGammar1}
Y(z) = \begin{pmatrix}
\frac1{2\pi i} e^{n(\tilde g(z)-\tilde\ell)}\oint_{\Gamma_{r_1}} e^{n\tilde\phi(s)} h_{-\gamma}(s) \frac{ds}{s-z} & e^{-n(\tilde g(z)-\tilde\ell)}\\
-e^{n(\tilde g(z)-\tilde\ell)} & \frac1{2\pi i} e^{-n(\tilde g(z)-\tilde\ell)}\oint_{\Gamma_{r_2}} e^{-n\tilde\phi(s)} h_\gamma(s) \frac{ds}{s-z}
\end{pmatrix}\\
+ e^{n(\phi(r_1)-\phi(r_2))}\mathcal O\begin{pmatrix}
e^{n(\phi(r_1)+g(z)-\ell)} & e^{-n(g(z)-\ell)}\\
e^{n(g(z)-\ell)} & e^{-n(\phi(r_2)+g(z)-\ell)}
\end{pmatrix}
\end{multline}
as $n\to\infty$, uniformly for $z$ in compact subsets of $\operatorname{Int }\Sigma_{r_1}$. Note that
\begin{align*}
e^{n(\tilde g(z)-\tilde\ell)}= \frac1{(1-x^2z)^{\alpha+\frac\gamma2}}.
\end{align*}
for such $z$. Now let us look at $\operatorname{Ext }\Sigma_{r_2}$. Here we find
\begin{multline} \label{eq:behavYSigmaInftyUnbounded}
Y(z) = \begin{pmatrix}
h_{-\gamma}(z) z^n & -\frac1{2\pi i}h_\gamma(z) z^{-n} \oint_{\Gamma_{r_1}} e^{n\tilde\phi(s)} h_{-\gamma}(s) \frac{ds}{s-z}\\
\frac1{2\pi i} h_\gamma(z) z^n \oint_{\Gamma_{r_2}} e^{-n\tilde\phi(s)} h_\gamma(s) \frac{ds}{s-z} & h_\gamma(z) z^{-n}
\end{pmatrix}\\
+ e^{n(\phi(r_1)-\phi(r_2))}\mathcal O\begin{pmatrix}
z^n & z^{-n} e^{n\phi(r_1)}\\
z^n e^{-n\phi(r_2)} & z^{-n}
\end{pmatrix}
\end{multline}
as $n\to\infty$, uniformly for $z$ in compact subsets of $\operatorname{Ext }\Sigma_{r_2}$. 
%Next we focus on $z\in\Omega_1$. We see that
%\begin{multline}
%Y(z) = \\
%\begin{pmatrix}
%\frac1{2\pi i}\oint_{\Gamma_{r_1}} e^{n\phi(s)} h_{-\gamma}(s) \frac{ds}{s-z} - e^{n(g(z)-\phi(z)+\ell)} & e^{-n(g(z)-\ell)}\\
%-e^{n(g(z)-\ell)} - e^{n(g(z)-\phi(z)+\ell)} \frac1{2\pi i}\oint_{\Gamma_{r_2}} e^{-n\phi(s)} h_\gamma(s) \frac{ds}{s-z}  & \frac1{2\pi i}\oint_{\Gamma_{r_2}} e^{-n\phi(s)} h_\gamma(s) \frac{ds}{s-z}
%\end{pmatrix}\\
%+ e^{n(\phi(r_1)-\phi(r_2))}\mathcal O\begin{pmatrix}
%e^{n(\phi(r_1)+g(z)-\ell)}+e^{n(g(z)-\phi(z)+\ell)} & e^{-n(g(z)-\ell)}\\
%e^{n(g(z)-\ell)}+e^{n(g(z)-\phi(z)+\ell-\phi(r_2))} & e^{-n(\phi(r_2)+g(z)-\ell)}
%\end{pmatrix}
%\end{multline}
%as $n\to\infty$, uniformly for $z$ in compact subsets of $\Sigma_1$. Lastly, let us focus on $z\in\Omega_2$. We see that
%\begin{multline}
%Y(z) = \\
%\begin{pmatrix}
%h_{-\gamma}(z) z^k 
%-\frac1{2\pi i} e^{n(g(z)+\phi(z))} \oint_{\Gamma_{r_1}} e^{n\phi(s)} h_{-\gamma}(s) \frac{ds}{s-z}
%& -\frac1{2\pi i} h_\gamma(z) z^{-k} \oint_{\Gamma_{r_1}} e^{n\phi(s)} h_{-\gamma}(s) \frac{ds}{s-z}\\
%\frac1{2\pi i} h_{-\gamma}(z) z^k \oint_{\Gamma_{r_2}} e^{-n\phi(s)} h_\gamma(s) \frac{ds}{s-z} + e^{n(g(z)+\phi(z))} & h_\gamma(z) z^{-k}
%\end{pmatrix}\\
%+ e^{n(\phi(r_1)-\phi(r_2))}\mathcal O\begin{pmatrix}
%z^k + e^{n(\phi(r_1)+g(z)+\phi(z))} & z^{-k} e^{n\phi(r_1)}\\
%z^k e^{-n\phi(r_2)} + e^{n(g(z)+\phi(z))} & z^{-k}
%\end{pmatrix}
%\end{multline}
%as $n\to\infty$, uniformly for $z$ in compact subsets of $\Sigma_2$.\\

Due to the freedom we have in picking $r$, $r_1$ and $r_2$, the regions $\Omega_1$ and $\Omega_2$ are not important, at least not when we are satisfied with exponentially small error terms that may not be sharp. For example, for any $t\in (1,z_0]$ we have by \eqref{eq:behavYSigmaInftyUnbounded} that
\begin{align} \label{eq:Y11ExtGammar2}
P_k(z) = Y_{11}(z) = \left(\left(\frac{z}{z-1}\right)^{\frac{\gamma}{2}}  + \mathcal O(e^{n(\phi(t')-\phi(t))})\right) z^k,
\end{align}
as $n\to\infty$, uniformly for $z\in \operatorname{Ext }\Gamma_{t}$ and the error is exponentially small provided that $t'\in(1,t)$ (because $\phi$ is increasing on $(1,z_0)$). On the other hand, for $t\in(1,z_0)$ close enough to $1$, we have by \eqref{eq:YinIntGammar1} that
\begin{align} \nonumber
Y_{11}(z) =& \left(\frac1{2\pi i}\oint_{\Gamma_{r_1}} e^{n\phi(s)} (1-x^2s)^\frac{\gamma}{2} \left(\frac{s}{s-1}\right)^{\frac{\gamma}{2}}  \frac{ds}{s-z}+ \mathcal O(e^{k(2\phi(t)-\phi(z_0))})\right)\\ \label{eq:Y11IntGammar1} 
&\times  \frac1{(1-x^2z)^{\alpha+\frac\gamma2}}
\end{align}
as $n\to\infty$, uniformly for $z\in\operatorname{Int }\Gamma_{t}$. In the next section we analyze the remaining integral in \eqref{eq:Y11IntGammar1}. Its behavior is dominated by an $e^{n\phi(1)}$ term. Indeed, it is possible to find $t\in (1,z_0)$ such that $2\phi(t)-\phi(z_0)<\phi(1)$. 
Clearly $\phi(t)$ is smaller than the average of $\phi(1)$ and $\phi(z_0)$ for $t>1$ sufficiently close to $1$. The region $\Sigma_\infty$ is thus sufficient to get asymptotics on any region in $\mathbb C$. Note that, since the convergence is uniform on compact sets of the regions above, we may differentiate the entries and get estimates for those as well, which we shall need in Section \ref{se:aadi}.\\

Naively, one may suspect that we can deform $\Gamma_{r_1}$ in \eqref{eq:Y11IntGammar1} to pass through the saddle point $z_0$, but the steepest descent directions of $\phi$ are parallel to the real line (see \eqref{eq:phiatz0}), and this poses a problem (any deformation must intersect the level curve $\phi(z)=\phi(z_0)$ in a point $z\neq z_0$). On the other hand, the steepest descent directions at $z_0$ of $-\phi$ are perpendicular to the real line, and here we may indeed apply the method of steepest descent to conclude that
\begin{multline} \nonumber
\oint_{\Gamma_{r_2}} e^{-n\phi(s)} (1-x^2s)^{-\frac\gamma2}h_{\gamma}(s)  \frac{ds}{s-z}\\ \label{eq:steepestDescentGamma2}
= (1+\mathcal O(n^{-1})) \sqrt\frac{2\pi}{-\phi''(z_0) n} e^{-n\phi(z_0)} \left(\frac{z_0-1}{z_0}\right)^{\frac{\gamma}{2}}  \frac{(1-x^2z_0)^{-\frac\gamma2}}{z_0-z}
\end{multline}
as $n\to\infty$, uniformly for $z$ some finite distance away from $z_0$. Note that $-\phi''(z_0)=\frac1c(c+1)^3x^4>0$. The case that $z$ is close to $z_0$ can also be treated and yields a contribution of order $1$ as $n\to\infty$, e.g., see \cite{Molag2023Edge} (Proposition 3.1). 

\subsection{Analysis of $Y_{11}$ in a neighborhood of $\Gamma_1$} \label{sec:ContourInt}
 
Let $t>1$ be close to $1$ as described in the previous section, and consider \eqref{eq:Y11IntGammar1}.
To obtain the asymptotic behavior of $Y_{11}$ on compact subsets of $\operatorname{Int }\Gamma_t$ we are tasked with analyzing an integral of the form
\begin{align*}
\oint_{\Gamma_{t}} e^{n\tilde\phi(s)}h_{-\gamma}(s) \frac{ds}{s-z}
= \oint_{\Gamma_{t}} e^{n\phi(s)}(1-x^2s)^\frac{\gamma}{2}\left(\frac{s}{s-1}\right)^{\frac{\gamma}{2}} \frac{ds}{s-z}.
\end{align*}
Here we allow $\gamma$ to be complex. We may deform $\Gamma_t$ to any contour $\gamma_t$ that encloses $(0,1]$, and we assume that it passes through $0$ and does not pass through $1$. Note that 
\begin{align*}
\re\phi'(1) = 1-\frac{c x^2}{1-x^2} = \frac{1-(1+c) x^2}{1-x^2}>0,
\end{align*}
since $x^2(1+c)<1$ in the strong regime. 
%More generally
%\begin{align*}
%\phi^{(m)}(1) = (m-1)!(-1)^{m-1}-c (m-1)!\frac{x^{2m}}{(1-x^2)^m}.
%\end{align*}
\begin{lemma}\label{lem:ointsnue-ksswExpansion}
Let $p>\re\frac{\gamma}{2}$ be a positive integer. Then there exist analytic functions $c_0(z), c_1(z), \ldots, c_{p-1}(z)$ on $\mathbb C\setminus\{\gamma_{t}\}$ such that as $n\to\infty$
\begin{multline*}
\oint_{\Gamma_{t}} e^{n\tilde\phi(s)} h_{-\gamma}(s) \frac{ds}{s-z}\\
= \sum_{m=0}^{p-1} c_m(z) e^{n\tilde\phi(1)}\oint_{\gamma_\infty} s^{m-\frac{\gamma}{2}} e^{-n s} \frac{ds}{s+\phi(z)-\phi(1)}
+\mathcal O(e^{n\phi(1)} n^{\re\frac{\gamma}{2}-p}).
\end{multline*}
uniformly for $z\not\in \gamma_{t}$, where $\gamma_\infty$ is a closed contour going around $[0,\infty)$ with negative orientation that passes through $+\infty$, and $s^{m-\frac{\gamma}{2}}$ has a cut at $[0,\infty)$. Furthermore, we have
\begin{align*}
c_0(z) = 
\begin{cases} -(-\phi'(1))^{\frac{\gamma}{2}-1} \frac{\phi(z)-\phi(1)}{z-1}, & z\neq 1\\
	(-\phi'(1))^{\frac{\gamma}{2}}, & z=1. \end{cases}
\end{align*} 
\end{lemma}
\begin{proof} Deform $\Gamma_t$ to $\gamma_t$ as described above. 
Let $0<\delta<z_0-1$. Let $\gamma_{t+}$ be the part of $\gamma_{t}$ in the disc $|s-1|\leq \delta$, and let $\gamma_{t-}$ be the remaining part of $\gamma_{t}$. We may assume without loss of generality (possibly having chosen $\gamma_{t}$ tighter around $[0,1]$ when opening lenses) that $\re \phi(s)<\phi(1)-\epsilon$ on $\gamma_{t-}$ for some $\epsilon>0$. We thus have
\begin{multline*}
\oint_{\gamma_{t}} \left(\frac{s}{s-1}\right)^\frac{\gamma}{2} (1-x^2s)^\frac{\gamma}{2} e^{n\phi(s)} \frac{ds}{s-z}\\
= \int_{\gamma_{t+}} \left(\frac{s}{s-1}\right)^\frac{\gamma}{2} (1-x^2s)^\frac{\gamma}{2} e^{n\phi(s)} \frac{ds}{s-z}+\mathcal O(e^{n\phi(1)-n \epsilon})
\end{multline*}
as $n\to\infty$, uniformly for $z\not\in \gamma_{t}$. Now let us define 
\begin{align*}
\psi(s)=\phi(1)-\phi(s). 
\end{align*}
From Lemma \ref{lem:phiAnalyticIso} we know that $\psi$ is an analytic isomorphism from $(\overline{\mathcal S_1}\cup \mathcal S_2 \cup \mathcal S_2^*)\setminus (-\infty,0]$ onto the strip $\{z\in \mathbb C : |\operatorname{Im} z|<\pi\}$.
In particular, we can restrict it to a conformal map in an open neighborhood of $s=1$ containing $\gamma_{t+}$. Note that $\psi(1)=0$. Substitution of integration variables yields
\begin{multline*}
e^{-n\phi(1)}\oint_{\gamma_{t+}} \left(\frac{s}{s-1}\right)^\frac{\gamma}{2} (1-x^2s)^\frac{\gamma}{2} e^{n\phi(s)} \frac{ds}{s-z}\\
=\oint_{\psi(\gamma_{t+})} \left(\frac{\psi^{-1}(s)}{\psi^{-1}(s)-1}\right)^{\frac{\gamma}{2}}
(1-x^2\psi^{-1}(s))^\frac{\gamma}{2} e^{n s} \frac{(\psi^{-1})'(s)}{\psi^{-1}(s)-z} ds\\
= \oint_{\psi(\gamma_{t+})} s^{-\frac{\gamma}{2}} e^{-n s} f(s) \frac{ds}{s-\psi(z)},
\end{multline*}
where $f$ is defined by
\begin{align*}
f(s) = \left(\frac{\psi^{-1}(s)-1}{\psi^{-1}(s)}\right)^{-\frac{\gamma}{2}} s^{\frac\gamma2} (1-x^2\psi^{-1}(s))^\frac{\gamma}{2} \frac{s-\psi(z)}{\psi^{-1}(s)-z} (\psi^{-1})'(s).
\end{align*}
Note that $f$ is an analytic function on an open neighborhood containing $\psi(\gamma_{t+})$. Hence $f$ has an absolutely convergent power series
\begin{align*}
f(s) = \sum_{m=0}^\infty c_m(z) s^m. 
\end{align*}
Using the inverse function theorem one can show that
\begin{align*}
c_0(z) = \lim_{s\to 0} f(s) = 
\begin{cases} -(-\phi'(1))^{\frac{\gamma}{2}-1} \frac{\phi(z)-\phi(1)}{z-1} (1-x^2)^\frac{\gamma}{2}, & z\neq 1\\
	(-\phi'(1))^{\frac{\gamma}{2}}(1-x^2)^\frac{\gamma}{2}, & z=1. \end{cases}
\end{align*} 
The factor $(1-x^2)^\frac{\gamma}{2}$ can be absorbed into $\phi$ to create $\tilde\phi$. 
The integrand of the integral
\begin{align*}
\oint_{\psi(\gamma_{t+})} s^{-\frac{\gamma}{2}} e^{-n s} \left(f(s) - \sum_{m=0}^{p-1} c_m s^m\right) \frac{ds}{s-\psi(z)} 
\end{align*}
is of order $\mathcal O(s^{p-\frac{\gamma}{2}})$ as $s\to 0$, and we may thus deform the integration contour to one that passes through $s=0$. Then Laplace's method tells us that
\begin{align*}
\oint_{\psi(\gamma_{t+})} s^{-\frac{\gamma}{2}} e^{-n s} \left(f(s) - \sum_{m=0}^{p-1} c_m s^m\right) \frac{ds}{s-\psi(z)} = 
\mathcal O(n^{\re\frac{\gamma}{2}-p})
\end{align*}
as $n\to\infty$. We still have $p$ remaining terms of the form 
\begin{align*}
\oint_{\psi(\gamma_{t+})} s^{-\frac{\gamma}{2}+m} e^{-n s} \frac{ds}{s-\psi(z)}.
\end{align*}
Extending the integration contour $\psi(\gamma_{t+})$ from both end points to $+\infty$, we get a closed contour $\gamma_\infty$ passing through $+\infty$. The difference between the integrals over $\psi(\gamma_{t+})$ and $\gamma_\infty$ is exponentially small. 
\end{proof}

\begin{lemma} \label{lem:ointsnue-kssw}
Let $\gamma_\infty$ be any closed loop with negative orientation enclosing $[0,\infty)$ and passing through $+\infty$. Let $\nu\in\mathbb C$ and assume that $w\not\in \gamma_{\infty}$.  Then we have for $k>0$
\begin{multline}
\oint_{\gamma_\infty} s^{\nu-1} e^{-k s} \frac{ds}{s-w}
= 2\pi i w^{\nu-1} e^{-k w} 
\begin{cases} \frac{\gamma(1-\nu, -k w)}{\Gamma(1-\nu)} & w\in \operatorname{Int} \gamma_{\infty} \\ \frac{\Gamma(1-\nu, -k w)}{\Gamma(1-\nu)}, & w\in \operatorname{Ext} \gamma_{\infty}\end{cases}.
\end{multline}
Here $\gamma(\cdot, \cdot)$ and $\Gamma(\cdot, \cdot)$ denote the lower and upper incomplete gamma function (defined with principal branch). 
\end{lemma}

\begin{proof}
Under the given conditions the integral determines a function that is analytic in both $\nu$ and $w$ (but $w\not\in\gamma_{\infty}$). We shall temporarily assume that $\re \nu>0$ and that $w\in \operatorname{Ext} \gamma_\infty$. In that case we can take the bandwidth of $\gamma_\infty$ to $0$ and get
\begin{align*}
\oint_{\gamma_\infty} s^{\nu-1} e^{-k s} \frac{ds}{s-w} &= (1-e^{-2\pi i\nu}) \int_0^\infty s^{\nu-1} e^{-k s} \frac{ds}{s-w}\\
&= (1-e^{-2\pi i\nu}) (-w)^{\nu-1} \Gamma(\nu) e^{-k w} \Gamma(1-\nu, -k w)\\
&= 2\pi i w^{\nu-1} e^{-k w} \frac{\Gamma(1-\nu, -k w)}{\Gamma(1-\nu)}
\end{align*}
according to EH II 137(3) in \cite{gradshteyn2007table}. By a similar argument, taking into account an extra residue contribution at $s=w$ we find the formula when $w\in\operatorname{Int} \gamma_{\infty}$ and $w\not\in [0,\infty)$. By analytic continuation the two formulae actually hold for all $\nu\in\mathbb C$, and all $w\not\in \gamma_{\infty}$. 
\end{proof}

\begin{proposition} \label{prop:intGammar1}
Let $t\in(1,z_0)$. 
Let $\gamma\in\mathbb C$ and $z\not\in \Gamma_{t}$. 
\begin{itemize}
\item[(i)] When $|n (z-1)|\to\infty$ as $n\to\infty$ we have uniformly for $z\in \operatorname{Int} \Gamma_t$
\begin{align*}
\frac{1}{2\pi i}\oint_{\Gamma_{t}}e^{n\tilde\phi(s)} h_{-\gamma}(s) \frac{ds}{s-z} 
&=\left(\frac{z}{z-1}\right)^\frac{\gamma}{2} e^{n\tilde\phi(z)}\\
&\quad + \frac{1}{z-1} \frac{1}{\Gamma(\gamma/2)} \phi'(1)^{\frac{\gamma}{2}} n^{\frac{\gamma}{2}-1} e^{n\tilde\phi(1)} (1+\mathcal O(n^{-1})).
\end{align*}
\item[(ii)] When $z=1-a/(\phi'(1) n)$ and $a\not\in [0,\infty)$ we have uniformly for $a\in\mathbb C$ in compact sets that as $n\to\infty$
\begin{align*}
\frac{1}{2\pi i}\oint_{\Gamma_{t}}e^{n\tilde\phi(s)} h_{-\gamma}(s) \frac{ds}{s-z} 
= e^{\pi i\gamma/2} (\phi'(1) n)^{\frac{\gamma}{2}} a^{-\frac{\gamma}{2}} e^{-a} \frac{\Gamma(\frac{\gamma}{2}, a)}{\Gamma(\frac{\gamma}{2})} e^{n\tilde\phi(1)} (1+\mathcal O(n^{-1})). 
\end{align*}
\item[(iii)] When $|n (z-1)|\to\infty$ as $n\to\infty$ we have uniformly for $z\in \operatorname{Ext} \Gamma_{t}$
\begin{align*}
\frac{1}{2\pi i}\oint_{\Gamma_{t}}e^{n\tilde\phi(s)} h_{-\gamma}(s) \frac{ds}{s-z} 
&= \frac{1}{z-1} \frac{1}{\Gamma(\gamma/2)} \phi'(1)^{\frac{\gamma}{2}} n^{\frac{\gamma}{2}-1} e^{n\tilde\phi(1)} (1+\mathcal O(n^{-1})).
\end{align*}
\end{itemize}
\end{proposition}

\begin{proof}
Let us first consider case (iii), i.e., we assume $z\in \operatorname{Ext} \Gamma_{t}$. Taking $\nu=-\frac{\gamma}{2}+m+1$ and $w=\phi(1)-\phi(z)$ in Lemma \ref{lem:ointsnue-kssw} and applying this to the expansion in Lemma \ref{lem:ointsnue-ksswExpansion}, we get that
\begin{multline} \label{eq:lem:ointsnue-kssw}
\frac{1}{2\pi i}\oint_{\Gamma_{t}}e^{n\tilde\phi(s)} h_{-\gamma}(s) \frac{ds}{s-z} \\
= e^{n\tilde\phi(1)}\sum_{m=0}^{p-1} c_m(z) w^{-\frac{\gamma}{2}+m} e^{-n w} \frac{\Gamma(\frac{\gamma}{2}-m, -n w)}{\Gamma(\frac{\gamma}{2}-m)}
+\mathcal O(e^{n\phi(1)} n^{\re\frac{\gamma}{2}-M}).
\end{multline}
We may then use the well-known asymptotic expansion
\begin{align*}
e^{-n w} \Gamma(\frac{\gamma}{2}-m, -n w) \sim (-n w)^{\frac{\gamma}{2}-m-1} \sum_{j=0}^\infty \frac{\Gamma(\frac{\gamma}{2}-m)}{\Gamma(\frac{\gamma}{2}-m-j)} (-n w)^{-j}, \quad |n w|\to\infty.
\end{align*}
Truncating the series after $m=0$ and $j=0$, we thus obtain
\begin{align*}
\frac{1}{2\pi i}\oint_{\Gamma_{t}}e^{n\tilde\phi(s)} h_{-\gamma}(s) \frac{ds}{s-z} 
&= -e^{n\tilde\phi(1)} \frac{c_0(z)}{w} e^{\pi i\gamma/2} \frac{1}{\Gamma(\frac{\gamma}{2})} n^{\frac{\gamma}{2}-1} \left(1+\mathcal O(n^{-1})\right)\\
&= \frac{1}{z-1} \frac{1}{\Gamma(\gamma/2)} \phi'(1)^{\frac{\gamma}{2}} n^{\frac{\gamma}{2}-1} e^{n\tilde\phi(1)} (1+\mathcal O(n^{-1}))
\end{align*}
as $n\to\infty$. Case (i) follows by a similar argument, now with the lower incomplete gamma function. Case (ii) follows from \eqref{eq:lem:ointsnue-kssw} and the observation that
\begin{align*}
w = \phi(1) - \phi\left(1-a/(\phi'(1) n\right) = \frac{a}{n}+\mathcal O(1/n^2), \qquad n\to\infty.
\end{align*}
\end{proof}

Notice that the condition $\phi(z)=\phi(1)$ corresponds to $z\in \Gamma_1$. This means that we can argue that $\phi(z)<\phi(1)$ when $z\in\operatorname{Int} \Gamma_1$. We are now ready for the proof of Theorem \ref{thm:mainPoly}.

\begin{proof}[Proof of Theorem \ref{thm:mainPoly}] 
When $z\in \operatorname{Ext }\Gamma_1\setminus U$ the result is already given in \eqref{eq:Y11ExtGammar2}. Now assume that $z\in\operatorname{Int }\Gamma_1\setminus U$. Since $U$ is assumed to be a thin neighborhood of $\Gamma_1$, we can find a $t>1$ close enough to $1$ such that $\operatorname{Int }\Gamma_1\setminus U\subset \operatorname{Int }\Gamma_t$. Then \eqref{eq:Y11IntGammar1} together with Proposition \ref{prop:intGammar1}(i) yields that
\begin{multline*}
P_n(z) = \frac1{(1-x^2z)^{\alpha+\frac\gamma2}}\\
\left(\left(\frac{z}{z-1}\right)^{\frac{\gamma}{2}} e^{n\tilde\phi(z)}
+ \frac{1}{z-1} \frac{1}{\Gamma(\gamma/2)} \phi'(1)^{\frac{\gamma}{2}} n^{\frac{\gamma}{2}-1} e^{n\tilde\phi(1)} (1+\mathcal O(n^{-1}))
+ \mathcal O(e^{n\phi(1)-n\epsilon})\right)
\end{multline*}
as $n\to\infty$, uniformly for $z\in \operatorname{Int }\Gamma_t$, for some constant $\epsilon>0$. We used here that $\alpha=c n$. In the case that $z\in \operatorname{Int }\Gamma_1\setminus U$, we have $|e^{n\phi(z)}|<e^{n\phi(1)}$ and the second term dominates, i.e., we have as $n\to\infty$
\begin{align*}
P_n(z) = \frac{1}{z-1} \frac{1}{\Gamma(\gamma/2)} \phi'(1)^{\frac{\gamma}{2}} n^{\frac{\gamma}{2}-1} e^{n\tilde\phi(1)} (1+\mathcal O(e^{-n\epsilon}))
\end{align*}
for some possibly different $\epsilon>0$. Plugging in the value for $\phi'(1)$ and using that $\tilde\phi(z)=\log z+(c+\frac\gamma{2n})\log(1-x^2z)$ we obtain the result in this region. If on the other hand $z\in U\setminus D_1(\delta)$ then both terms are of comparable order. Lastly, it follows from Proposition \ref{prop:intGammar1}(ii) that
\begin{multline*}
P_n(z) = \frac1{(1-x^2z)^{\alpha+\frac\gamma2}} \\
\left(e^{\pi i\gamma/2} (1-z)^{-\frac{\gamma}{2}} e^{n\phi'(1)(z-1)} \frac{\Gamma(\frac{\gamma}{2}, n\phi'(1)(1-z))}{\Gamma(\frac{\gamma}{2})} e^{n\tilde\phi(1)} (1+\mathcal O(n^{-1}))+ \mathcal O(e^{n\phi(1)-n\epsilon})\right)
\end{multline*}
as $n\to\infty$, uniformly for $|n (z-1)|=\mathcal O(1)$. That the convergence is uniform for $x\in(0,\sqrt{\tilde\mu})$ and $\re \gamma>-2$ in compact sets follows by a similar argument as in the proof of Theorem \ref{th:trunc-asympt} below. 
\end{proof}

\section{Asymptotic analysis of the differential identity}
\label{se:aadi}
In this section we complete the proof of Theorem \ref{th:trunc-asympt}. We recall the differential identity \eqref{diffid}, which we re-write in terms of the matrix entries of $Y$:
\begin{equation} \label{eq:aadi}
\begin{split}
&\frac{d}{dx}(\log R_{\gamma}(x)) = -2x\left(\frac{\gamma}{2}+\alpha\right)\left(-n u\right.\\
&\left.+ u^3 \left[\partial_z Y_{11,n}(z)\right]_{z=u}Y_{22,n+1}(u)-u^2 \left[\partial_z Y_{21,n+1}(z)\right]_{z=u}Y_{12,n}(u)\right)
\end{split}
\end{equation}
where $u=x^{-2}$. The asymptotic behavior of $Y(z)$ near $z=x^{-2}$ is determined by \eqref{eq:behavYSigmaInftyUnbounded}, where we shall make the choice $r_2=z_0$ and $r_1>1$ close to $1$ in this entire section. Since the asymptotics of $Y$ are uniform on any compact neighborhood of $z=x^{-2}$ (not intersecting with $\Gamma_r$), we may interchange limit and differentiation and find from \eqref{eq:behavYSigmaInftyUnbounded} for any fixed $0< x<\sqrt{\tilde\mu}$ that 
\begin{align*}
\left[\partial_z Y_{11,n}(z)\right]_{z=x^{-2}} &= x^{-2n+2} \left( \left(n-\frac\gamma2 \frac{x^2}{1-x^2}\right) (1-x^2)^{\frac\gamma2} + \mathcal O(n e^{n(\phi(r_1)-\phi(z_0))})\right)\\
Y_{22,n+1}(x^{-2}) &= x^{2n+2} \left((1-x^2)^{-\frac\gamma2}+\mathcal O(e^{n(\phi(r_1)-\phi(z_0))})\right) 
\end{align*}
as $n\to\infty$. Therefore
\begin{align}
-n u + u^3 \left[\partial_z Y_{11,n}(z)\right]_{z=u}Y_{22,n+1}(u)
= -\frac\gamma2 \frac1{1-x^2}+\mathcal O(n e^{n(\phi(r_1)-\phi(z_0))}).
\end{align}
We will show that the remaining terms in the differential identity \ref{eq:aadi} are negligible. Inserting \eqref{eq:steepestDescentGamma2} into the $21$-entry in \eqref{eq:behavYSigmaInftyUnbounded}, we have as $n\to\infty$
\begin{align*}
\left[\partial_z Y_{21,n+1}(z)\right]_{z=x^{-2}} &= \mathcal O\left(\sqrt{n} x^{-2n} e^{-n \phi(z_0)}+n x^{-2n+2} e^{n(\phi(1)-2\phi(z_0))}\right).
\end{align*}
On the other hand, using Proposition \ref{prop:intGammar1}(iii) and \eqref{eq:behavYSigmaInftyUnbounded}, we know that
\begin{align} \label{eq:Yx-22phi1phiz0}
Y_{12,n}(x^{-2}) = \mathcal O(n^{\frac\gamma2-1} x^{2n} e^{n\phi(1)}+x^{2n} e^{n(2\phi(r_1)-\phi(z_0))}),
\end{align}
as $n\to\infty$, and the first term is dominant when $r_1$ is close to $1$.
We thus conclude that as $n\to\infty$
\begin{align*}
\left[\partial_z Y_{21,n+1}(z)\right]_{z=x^{-2}} Y_{12,n}(x^{-2})
= \mathcal O(n^{\frac{\gamma-1}{2}} e^{n(\phi(1)-\phi(z_0))})
\end{align*}
and this is exponentially small. 

\begin{proof}[Proof of Theorem \ref{th:trunc-asympt}]
By the preceeding asymptotics, we find for any fixed $x\in (0,\sqrt{\tilde\mu})$ that 
\begin{equation}
\begin{split}
\frac{d}{dx}(\log R_{\gamma}(x)) = \frac{\gamma}{2}\left(\frac{\gamma}{2}+\alpha\right)\frac{2x}{1-x^{2}}
+\mathcal O(n e^{n(\phi(r_1)-\phi(z_0))}).
\end{split}
\end{equation}
as $n\to\infty$. We conclude that the expression converges pointwise as $n\to\infty$. So far, we have been assuming that $x$ is fixed. To get a uniform limit, we have to argue how the implied constant(s) in \eqref{eq:behavYSigmaInftyUnbounded} depends on $x$ and $\gamma$. Let us first investigate the case where $\gamma=0$. It is a well-known fact about Cauchy-operators that the implied constant depends on the length of the path. In other words, for $\gamma=0$, we have
\begin{equation}
\begin{split}
\frac{d}{dx}(\log R_{\gamma}(x)) &= \frac{\gamma}{2}\left(\frac{\gamma}{2}+\alpha\right)\frac{2x}{1-x^{2}}+\mathcal O(x^\sigma L(\Sigma_{r_1}) L(\Sigma_{r_2}) n e^{n(\phi(r_1)-\phi(z_0))}),
\end{split}
\end{equation}
where $L(\Sigma_{r_1})$ and $L(\Sigma_{z_0})$ denote the arc length of $\Sigma_{r_1}$ and $\Sigma_{z_0}$ respectively (which depend implicitly on $x$), $\sigma\in\mathbb R$ is some constant, and the constant implied by the $\mathcal O$-term does not depend on $x$. We are allowed to wind $\Sigma_{r_1}$ as tightly as we want around $[0,1]$ (letting $r_1$ be close enough to $1$) and have, say, $L(\Sigma_{r_1})\leq 3$ uniformly for $x\in [0,\sqrt{\tilde\mu}-\delta]$. Due to continuity in $x$, we may also assume that the condition $2\phi(r_1)-\phi(z_0)<\phi(r_1)$ is met, such that the first term on the right-hand side in \ref{eq:Yx-22phi1phiz0} is dominant. For the argument that follows, we may effectively set $\Sigma_{z_0}=\Gamma_{z_0}$. Remember that it is determined by
\begin{align*}
\re \phi(z) = \phi(z_0).
\end{align*}
A substitution $z=x^{-2} \zeta$ turns this into the condition
\begin{align*}
\log|\zeta|+c\log|1-\zeta|=c\log c - (c+1)\log(c+1).
\end{align*}
Hence, the length of $x^2\Gamma_{z_0}$ is independent of $x$, and we conclude that $L(\Gamma_{z_0})= C x^{-2}$ for some constant $C>0$. 
Now let us investigate $\phi(r_1)-\phi(z_0)$. We notice that
\begin{align} \label{eq:forMeanValue}
\phi(z_0)-\phi(r_1)=\int_{r_1}^{z_0} \left(\frac1t-\frac{c x^2}{1-x^2t}\right) dt
= \int_{x^2r_1}^{\frac1{c+1}} \frac{1-(c+1)t}{t(1-t)} dt,
\end{align}
which is positive when $r_1>1$ is close enough to $1$, uniformly for $x\in[0,\sqrt{\tilde\mu}-\delta]$. Namely, we need that $r_1<z_0$ and under the present conditions
\begin{align*}
z_0 = \frac1{(c+1)x^2}\geq x^{-2}\geq (\sqrt{\tilde\mu}-\delta)^{-2}
\end{align*}
and the right-hand side is bigger than $1$ and independent of $x$. (e.g., we could pick $r_1=\tilde\mu^{-1}$.) 
Furthermore, we notice that as $x\downarrow 0$
\begin{align*}
\phi(r_1)-\phi(z_0) = 2\log x+\mathcal O(1).
\end{align*}
We conclude that for $r_1>1$ close enough to $1$ (and independent of $x$)
\begin{equation} \label{eq:logRgammax=O}
\begin{split}
\frac{d}{dx}(\log R_{\gamma}(x)) = \frac{\gamma}{2}\left(\frac{\gamma}{2}+\alpha\right)\frac{2x}{1-x^{2}}
+\mathcal O(x^{\sigma-2} n e^{n(\phi(r_1)-\phi(z_0))})
\end{split}
\end{equation} 
as $n\to\infty$, uniformly for $x\in[0,\sqrt{\tilde\mu}-\delta]$, for some constant $\sigma\in\mathbb R$, and the constant implied by the $\mathcal O$-term does not depend on $x$. The general case $\re \gamma>-2$ is essentially the same, one merely has to add some factors to the $\mathcal O$-term corresponding to the maxima of $|1-1/s|^{- \frac\gamma2}$ and $|1-1/s|^{\frac\gamma2}$ on respectively $\Gamma_{r_1}$ and $\Gamma_{z_0}$, which depend continuously on $x$ (even around $x=0$) and $\gamma$ and are thus uniformly bounded for $x$ and $\gamma$ in compact sets. Thus \eqref{eq:logRgammax=O} holds uniformly for compact sets of $\re \gamma>-2$ as well and integrating \eqref{eq:logRgammax=O} gives
\begin{align*} 
\log R_{\gamma}(x)=\log R_{\gamma}(0) -\frac{\gamma}{2}\left(\frac{\gamma}{2}+\alpha\right) \log(1-x^2) 
+\mathcal O(x^{\sigma-1} n e^{n(\phi(r_1)-\phi(z_0))})
\end{align*}
as $n\to\infty$, uniformly for $x\in[0,\sqrt{\tilde\mu}-\delta]$, where the form of the error term follows by application of the mean value theorem and some trivial estimate of \eqref{eq:forMeanValue}. We thus conclude that
\begin{align}
R_{\gamma}(x)=R_{\gamma}(0) \frac{1+ \mathcal O(e^{-n\epsilon})}{(1-x^2)^{\frac{\gamma}{2}\left(\frac{\gamma}{2}+\alpha\right)}} \label{concthm}
\end{align}
as $n\to\infty$, for some constant $\epsilon>0$, uniformly for compact subsets of $x\in[0,\sqrt{\tilde\mu})$ and $\re\gamma>-2$. It remains to compute the asymptotics of the constant $R_{\gamma}(0)$. For $x=0$, the underlying two dimensional measure is radially symmetric, allowing for the explicit computation
\begin{equation}
R_{\gamma}(0) = \prod_{j=0}^{n-1}\frac{\Gamma(\frac{\gamma}{2}+j+1)\Gamma(j+\alpha+1)}{\Gamma(j+1)\Gamma(\frac{\gamma}{2}+j+\alpha+1)}. \label{rgam0}
\end{equation}
The asymptotics of such products is straightforward to calculate in terms of the Barnes G-function $G(z)$, see \cite{DS22}[Eq. 3.28] for a similar calculation. We have the asymptotic estimate
\begin{align} \nonumber
R_{\gamma}(0) &= n^{-\frac{\gamma n}{2}} \frac{G(\frac\gamma2+n+1)}{G(1+\frac\gamma2) G(n+1)}\\
&=n^{\frac{\gamma^{2}}{8}}\mu^{n\frac{\gamma}{2}}(1-\mu)^{\alpha\frac{\gamma}{2}+\frac{\gamma^{2}}{8}}\,\frac{(2\pi)^{\frac{\gamma}{4}}}{G(1+\frac{\gamma}{2})}\,(1+\mathcal O(1/n)), \qquad n \to \infty, \label{rgam0_as}
\end{align}
uniformly in any compact subset of $\re\gamma >-2$, where the last step follows from the well-known asymptotic series for $\log G(z+1)$ as $|z|\to\infty$. Inserting \eqref{rgam0_as} into \eqref{concthm} completes the proof.
\end{proof}

%This makes integrating the differential identity rather straightforward. We integrate from $x=0$ and immediately get
%\begin{equation}
%\log R_{\gamma}(x) \sim \log R_{\gamma}(0)-\frac{\gamma}{2}\left(\frac{\gamma}{2}+\alpha\right)\log(1-x^{2})
%\end{equation}
%Exponentiating, we get the asymptotics of the partition function as
%\begin{equation}
%R_{\gamma}(x) \sim R_{\gamma}(0)(1-x^{2})^{-\alpha \frac{\gamma}{2}-\frac{\gamma^{2}}{4}}.
%\end{equation}
%These asymptotics are consistent with the case of $\gamma$ given by an even integer that were obtained in \cite[Theorem 1.3]{SS23}. It was shown in \cite[Theorem 1.2]{SS23} that
%\begin{equation}
%R_{\gamma}(1) = \prod_{j=N-n+1}^{N}\frac{\Gamma(j)\Gamma(j+\gamma)}{(\Gamma(j+\frac{\gamma}{2}))^{2}}. \label{rg1}
%\end{equation}
From the asymptotics of $R_{\gamma}(x)$, we obtain the central limit theorem stated in Corollary \ref{cor:clt}. Note that the proof only needs the asymptotics for $\gamma \in (-\epsilon,\epsilon)$ for some small $\epsilon>0$.
\begin{proof}[Proof of Corollary \ref{cor:clt}]
Let $\varphi(t)$ denote the moment generating function of $\log|\det(A-x)|$, i.e.
\begin{equation}
\varphi(t) = \mathbb{E}(e^{t\log|\det(A-x)|}) = \mathbb{E}\left(|\det(A-x)|^{t}\right).
\end{equation}
Hence, $\varphi(t) = R_{t}(x)$. Then we can obtain asymptotics of $\varphi(t)$ using \eqref{rgamx_as}. Denoting the moment generating function of the standardised quantity on the the left-hand side of \eqref{clt} by $\tilde{\varphi}(t)$, we apply the asymptotics in \eqref{rgamx_as} with $\gamma = \frac{t}{2\sqrt{\log n}}$. A straightforward computation shows that the only surviving term comes from the factor $n^{\frac{\gamma^{2}}{8}}$ in \eqref{rgam0_as}. The other terms are cancelled by the centering and scaling in \eqref{clt}. This results in the asymptotics
\begin{equation}
\tilde{\varphi}(t) = e^{\frac{t^{2}}{2}}\left(1+o(1)\right), \qquad n \to \infty,
\end{equation}
valid for any fixed $t \in \mathbb{R}$. The stated central limit theorem follows.
\end{proof}

\section{Double scaling regime for $x$ near $1$}
In this section, we consider the weak regime, that is, $N-n = \alpha$ where $\alpha$ is a fixed positive integer, and we take the double scaling regime
\begin{equation}\label{eq:dscaling}
x^2=1-\frac{v}{n},\qquad v>0.
\end{equation}
We consider $x>0$, so $x=\sqrt{1-\frac{v}{n}}<1$ is inside the unit disc. In this case, one needs a modification of the local parametrix, since the gap between the cuts $[0,1]$ and $[x^{-2},\infty)$ closes as $n\to\infty$. We follow the reference \cite{CIK11} to make a connection with solutions of the Painlev\'e V differential equation. 

The asymptotic results in \cite{CIK11} are derived for Toeplitz determinants of the form
\begin{equation}
D_n=\det\left(f_{j-k}\right)_{j,k=0}^{n-1}, \qquad f_j=\frac{1}{2\pi}\int_{0}^{2\pi}f(e^{i\theta})e^{-ij\theta}d\theta,
\end{equation}
where the symbol is given in \cite[(1.7)]{CIK11}:
\begin{equation}\label{eq:fs}
%f(z) = (z-e^{t})^{\alpha+\beta}(z-e^{-t})^{\alpha-\beta}z^{-\alpha+\beta}e^{-\pi i (\alpha+\beta)}e^{V(z)},
f(s;t) = (s-e^{t})^{a+b}(s-e^{-t})^{a-b}s^{-a+b}e^{-\pi i (a+b)}e^{V(s)}.
\end{equation}

In order to match this function with our weight, we take the point $1/x$, which lies in the gap between the two cuts $[0,1]$ and $[x^{-2},\infty)$ if $x<1$, and we make the change of variable $s=xz$ in \eqref{cweight}. It follows that $z=1$ is mapped to $s=x<1$, $z=x^{-2}$ is mapped to $s=1/x>1$ and $z=1/x$ is mapped to $s=1$. We obtain the weight 
\begin{equation}
\begin{aligned}
w(s;x)
&=
(xs-1)^{\alpha+\frac{\gamma}{2}}e^{-\pi i\left(\alpha+\frac{\gamma}{2}\right)}\left(\frac{s-x}{s}\right)^{\frac{\gamma}{2}}\\
&=
x^{\alpha+\frac{\gamma}{2}}(s-x^{-1})^{\alpha+\frac{\gamma}{2}}
e^{-\pi i\left(\alpha+\frac{\gamma}{2}\right)}
\left(\frac{s-x}{s}\right)^{\frac{\gamma}{2}},
\end{aligned}
\end{equation}
with cuts on $[0,x]$ and $[x^{-1},\infty)$.
% (\textcolor{magenta}{this is the same as saying that $\arg\, s, \arg\, (s-x), \arg\,(s-x^{-1})\in(0,2\pi)$, correct?}). 

If we write now $x=e^{-t}$ and we compare with \eqref{eq:fs}, we can identify the parameters and the potential as follows:
\begin{equation}
a=\frac{\gamma+\alpha}{2},\qquad
b=\frac{\alpha}{2},\qquad
V(s)=V_0=\left(\alpha+\frac{\gamma}{2}\right)\log s.
\end{equation}
Note that we use $a$ and $b$ for $\alpha$ and $\beta$ in reference \cite{CIK11}.
%(N: Interesting observation above. The problem in \cite{CIK11} definitely seems related. Does it matter that we have the $z^{k}$ term in our weight?)

With this connection, we can relate the Toeplitz determinant $D_n$ with $R_{\gamma}(x)$ directly:
%with is directly related with  our determinant \textcolor{red}{$T_M$}, and as a consequence we can write $R_{\gamma}(x)$ in terms of $D_M(x)$. 
%We have the following result:, taking $w(z)=f(z)$ from \eqref{eq:fs}
\begin{lemma}
The quantity $R_{\gamma}(x)$ given in \eqref{eq:Rgamma} satisfies
\begin{equation}\label{eq:RgammaDM}
\begin{aligned}
R_{\gamma}(x) 
%&= 
%\frac{M!}{Z_{M}}\,\left(\prod_{j=0}^{M-1}\frac{\Gamma(\frac{\gamma}{2}+j+1)\Gamma(\alpha)}{\Gamma(\frac{\gamma}{2}+j+\alpha+1)}\right)T_{M}\\
&=
%\prod_{j=0}^{M-1}\frac{\Gamma(\frac{\gamma}{2}+j+1)\Gamma(j+\alpha+1)}{j!\Gamma(\frac{\gamma}{2}+j+\alpha+1)}D_{M}.
\frac{G\left(\frac{\gamma}{2}+n+1\right)G(\alpha+n+1)}{G(n+1)G\left(\frac{\gamma}{2}+n+\alpha+1\right)}
D_{n}(x),
\end{aligned}
\end{equation}
in terms of the Barnes G function.
\end{lemma}

\begin{proof}
Taking $w(z)=f(z)$ from \eqref{eq:fs} and writing  $z=e^{i\theta}$, we obtain
\[
\begin{aligned}
T_n(x, \gamma)
=
\det\left(\int_{S^1} z^{j-k}w(z)\frac{dz}{2iz}\right)_{j,k=0}^{n-1}
&=
\det\left(\frac{1}{2}\int_{0}^{2\pi}  f(e^{i\theta}) e^{i(j-k)\theta}d\theta\right)_{j,k=0}^{n-1}\\
&=
\pi^n \det\left(f_{k-j}\right)_{j,k=0}^{n-1}=\pi^n D_n(x).
\end{aligned}
\]
%This results in %$T_M=\pi D_M^T$, so 
%$\det(T_M)=\pi^{M} D_M$. 
We recall the connection between $R_{\gamma}(x)$ and $T_n$, see \eqref{eq:EgammaTM} and \eqref{eq:Rgamma}, and the formula above gives 
\begin{equation}
\begin{aligned}
R_{\gamma}(x) 
&= 
\frac{n!}{Z_{N,n}}\,\left(\prod_{j=0}^{n-1}\frac{\Gamma(\frac{\gamma}{2}+j+1)\Gamma(\alpha)}{\Gamma(\frac{\gamma}{2}+j+\alpha+1)}\right)T_n(x)\\
&=
\prod_{j=0}^{n-1}\frac{\Gamma(\frac{\gamma}{2}+j+1)\Gamma(j+\alpha+1)}{j!\Gamma(\frac{\gamma}{2}+j+\alpha+1)}D_{n}(x),
\end{aligned}
\end{equation}
since 
\begin{equation}
Z_n
=
\pi^{n}\Gamma(\alpha)^{n}n!\prod_{j=0}^{n-1}\frac{j!}{\Gamma(j+\alpha+1)},
\end{equation}
from formula (5.32) in \cite{DS22}. This gives \eqref{eq:RgammaDM}, writing the product of Gamma functions as Barnes G functions.
\end{proof}

The next result, directly derived from Theorem 1.1 in \cite{CIK11}, gives large $n$ asymptotics for $\log D_n(t)$, provided that $\textrm{Re}\, a>-\frac{1}{2}$ and $b\in i\mathbb{R}$, uniformly for $0 < t < t_0$, for $t_0$ sufficiently small. This result is extended in Theorem 1.4 in \cite{CIK11} for $a,b\in\mathbb{C}$, with $\textrm{Re}\, a>-\frac{1}{2}$, $a\pm b\neq -1,-2,\ldots$ (which in our notation means $\alpha-1+\frac{\gamma}{2},\frac{\gamma}{2}\neq -1,-2,\ldots$). Recall that $x=e^{-t}$, so in our context, this corresponds to uniformity for 
$1-x^\ast< x <1$, for $x^\ast$ sufficiently small. 

The asymptotic approximation for the Toeplitz determinant $D_n(x)$ is written in terms of $\sigma(u)$, which is a solution of the Jimbo--Miwa--Okamoto $\sigma$-Painlev\'e V equation \eqref{eq:PV} with boundary conditions \eqref{eq:PVbcc}.
%:
%\begin{equation}\label{eq:sigmaPV}
%(u\sigma'')^2
%=
%\left(\sigma-u\sigma'+2(\sigma')^2+2a\sigma')\right)-4(\sigma')^2(\sigma'+a+b)(\sigma'+a-b),
%\end{equation}
%which is real analytic on $(0, +\infty)$, except for possible isolated poles, and has specific boundary conditions as $u\to 0$ and as $u\to\infty$: 
%%(depending on whether $2a$ is an integer or not) and as $u\to\infty$: 
%
%\begin{equation}\label{eq:asympsigma}
%\sigma(u)
%=
%\begin{cases}
%a^2-b^2+\frac{a^2-b^2}{2a}\left\{u+u^{1+2a}C(a,b)\right\}\left(1+\mathcal{O}(u)\right), & \,\, u\to 0, \,\,2a\notin\mathbb{Z}\\[2mm]
%a^2-b^2+\mathcal{O}(u)
%+
%\mathcal{O}(u^{1+2a})+
%\mathcal{O}(u^{1+2a}\log u), & \,\, u\to 0, \,\,2a\in\mathbb{Z}\\[2mm]
%\displaystyle -\frac{u^{-1+2a}e^{-u}}{\Gamma(a-b)\Gamma(a+b)}\left(1+\mathcal{O}(u^{-1})\right),&\,\, u\to+\infty.
%\end{cases}
%\end{equation}

We have the following result:

\begin{proposition}
For fixed values of $v\in \mathbb{R}^+$, bounded away from poles of the function $\sigma$, 
%In the variable $t=-\log x$, for $0 < t < t_0$, with $t_0$ sufficiently small, and 
%
and if $\alpha+\frac{\gamma}{2},\frac{\gamma}{2}\neq -1,-2,\ldots$, we have the following asymptotic behavior
\begin{equation}\label{eq:asympDt}
\begin{aligned}
D_n(v)
&=
%e^{\left(\frac{\gamma}{2}+\alpha-1\right)Mt}
\left(\frac{n}{v}\right)^{\frac{\gamma}{2}\left(\frac{\gamma}{2}+\alpha\right)}
\exp\left(-\int_{v}^{\infty} \frac{\sigma(u)}{u}du\right)
(1+o(1)),\qquad n\to\infty.
\end{aligned}
\end{equation}

\end{proposition}

\begin{proof}
From Theorem 1.1 in \cite{CIK11}, we have the asymptotic expansion
%\alpha-1+\frac{\gamma}{2}\neq -1,-2,\cdots$
\begin{equation}\label{eq:asymp_lnDM}
\begin{aligned}
\log D_n(t)
&=
n\left(\alpha-1+\frac{\gamma}{2}\right)\log x
+
(a+b) n t
+
\sum_{k=1}^{\infty}(a^2-b^2)\frac{e^{-2tk}}{k}
\\
&
+\log\frac{G(1+a+b)G(1+a-b)}{G(1+2a)}+\Omega(2nt)+o(1)\\
&=
%M\left(\alpha-1+\frac{\gamma}{2}\right)\log x
%+
%\left(\frac{\gamma}{2}+\alpha-1\right)Mt
-\frac{\gamma}{2}\left(\frac{\gamma}{2}+\alpha\right)\log(1-e^{-2t})\\
&+
\log\frac{G\left(\frac{\gamma}{2}+\alpha\right)G\left(1+\frac{\gamma}{2}\right)}{G(\alpha+\gamma+1)}
+\Omega(2nt)+o(1),
\end{aligned}
\end{equation}
where $G(z)$ is Barnes's G-function and 
\begin{equation}
\begin{aligned}
\Omega(2nt)
&=
\int_0^{2nt} \frac{\sigma(u)-a^2+b^2}{u}du +(a^2-b^2)\log(2nt)\\
&=
\int_0^{2nt} \frac{\sigma(u)-\frac{\gamma}{2}\left(\frac{\gamma}{2}+\alpha\right)}{u}du +\frac{\gamma}{2}\left(\frac{\gamma}{2}+\alpha\right)\log(2nt).
\end{aligned}
\end{equation}

Note that in our setting, $\alpha=N-n$ is an integer, so the boundary behavior as $u\to 0$ changes depending on $\gamma$ being an integer or not.

The authors in \cite[(1.31)]{CIK11} also observe that
\begin{equation}\label{eq:Omegainfty}
\Omega(+\infty)
=
-\log \frac{G(1+a+b)G(1+a-b)}{G(1+2a)}
=
-\log \frac{G\left(\alpha+\frac{\gamma}{2}+1\right)G\left(1+\frac{\gamma}{2}\right)}{G(\gamma+\alpha+1)},
\end{equation}
using the facts that $a+b=\alpha+\frac{\gamma}{2}$, $a-b=\frac{\gamma}{2}$ and $2a=\gamma+\alpha$.

We can use this result to rewrite the integral formula for $\Omega(2nt)$ as follows: fix $n$ and fix $R>2nt$, then
\begin{equation}
\begin{aligned}
\Omega(2nt)
&=
\int_0^{R} \frac{\sigma(u)-\frac{\gamma}{2}\left(\frac{\gamma}{2}+\alpha\right)}{u}du 
-
\int_{2nt}^{R} \frac{\sigma(u)-\frac{\gamma}{2}\left(\frac{\gamma}{2}+\alpha\right)}{u}du\\
&+\frac{\gamma}{2}\left(\frac{\gamma}{2}+\alpha-1\right)\log(2nt)\\
&=
\int_0^{R} \frac{\sigma(u)-\frac{\gamma}{2}\left(\frac{\gamma}{2}+\alpha\right)}{u}du 
+\frac{\gamma}{2}\left(\frac{\gamma}{2}+\alpha\right)\log(R)
-
\int_{2nt}^{R} \frac{\sigma(u)}{u}du\\
&=
\Omega(R)-\int_{2nt}^{R} \frac{\sigma(u)}{u}du.
\end{aligned}
\end{equation}
If we let $R\to +\infty$, we obtain
\begin{equation}
\begin{aligned}
\Omega(2nt)
&=
\Omega(+\infty)-\int_{2nt}^{\infty} \frac{\sigma(u)}{u}du\\
&=
%-\log \frac{G(1+a+b)G(1+a-b)}{G(1+2a)}-\int_{2Mt}^{\infty} \frac{\sigma(u)}{u}du.
-\log \frac{G\left(\alpha+\frac{\gamma}{2}+1\right)G\left(1+\frac{\gamma}{2}\right)}{G(\gamma+\alpha+1)}
-\int_{2nt}^{\infty} \frac{\sigma(u)}{u}du.
\end{aligned}
\end{equation}

%and therefore
%\begin{equation}
%\begin{aligned}
%\Omega(2Mt)
%&=
%\Omega(+\infty)-\int_{2Mt}^{\infty} \frac{\sigma(u)-\frac{\gamma}{2}\left(\frac{\gamma}{2}+\alpha-1\right)}{u}du -\frac{\gamma}{2}\left(\frac{\gamma}{2}+\alpha-1\right)\log(2Mt)\\
%&=
%\Omega(+\infty)-\int_{2Mt}^{\infty} \frac{\sigma(u)}{u}du -\frac{\gamma}{2}\left(\frac{\gamma}{2}+\alpha-1\right)\log(2Mt)\\
%\end{aligned}
%\end{equation}

%In particular, consider the case $\gamma=2m$, with $m\geq 1$, then $2a=\alpha-1+\gamma=N-n-1+2m$, and we have the boundary conditions
%\begin{equation}
%\sigma(u)
%=
%\begin{cases}
%\displaystyle m\left(m+\alpha-1\right)
%+
%\mathcal{O}(u)
%+
%\mathcal{O}(u^{\alpha+2m})+
%\mathcal{O}(u^{\alpha+2m}\log u), & \,\, u\to 0,\\[2mm]
%\displaystyle -\frac{u^{\alpha+2m-2}e^{-u}}{\Gamma(m)\Gamma(\alpha-1+m)}\left(1+\mathcal{O}(u^{-1})\right),&\,\, u\to+\infty.
%\end{cases}
%\end{equation}
%Also,
%\begin{equation}
%\begin{aligned}
%\Omega(2Mt)
%&=
%%-\log \frac{G(1+a+b)G(1+a-b)}{G(1+2a)}-\int_{2Mt}^{\infty} \frac{\sigma(u)}{u}du.
%%-\log \frac{G(\alpha)G\left(1+m\right)}{G(\alpha+2m)}
%-\log \frac{G\left(\alpha+m\right)G\left(1+m\right)}{G(\alpha+2m)}
%-\int_{2Mt}^{\infty} \frac{\sigma(u)}{u}du,
%%\int_0^{2Mt} \frac{\sigma(u)-m\left(m+\alpha-1\right)}{u}du 
%%+m\left(m+\alpha-1\right)\log(2Mt).
%\end{aligned}
%\end{equation}

As a consequence, inserting this in \eqref{eq:asymp_lnDM}, we obtain
\begin{equation}
\begin{aligned}
\log D_n(t)
&=
%(m+\alpha-1) Mt +m\left(m+\alpha-1\right)\log(1-e^{-2t})\\
%\left(\frac{\gamma}{2}+\alpha-1\right)Mt
-\frac{\gamma}{2}\left(\frac{\gamma}{2}+\alpha\right)\log(1-e^{-2t})\\
&+
\log\frac{G\left(\frac{\gamma}{2}+\alpha+1\right)G\left(1+\frac{\gamma}{2}\right)}{G(\alpha+\gamma+1)}
+\Omega(2nt)+o(1)\\
%
%&+
%\log\frac{G\left(m+\alpha\right)G\left(1+m\right)}{G(\alpha+2m)}
%+\Omega(2Mt)+o(1)\\
&=
%\left(\frac{\gamma}{2}+\alpha-1\right)Mt
-\frac{\gamma}{2}\left(\frac{\gamma}{2}+\alpha\right)\log(1-e^{-2t})
%(m+\alpha-1) Mt+m\left(m+\alpha-1\right)\log(1-e^{-2t})\\
-
\int_{2nt}^{\infty} \frac{\sigma(u)}{u}du
+o(1).
\end{aligned}
\end{equation}
Note that $1-e^{-2t}=1-x^2=\frac{v}{n}$, using \eqref{eq:dscaling}, and also %$2Mt= -2M\log x = v+\mathcal{O}(M^{-1})$, since 
\begin{equation}\label{eq:txv}
2nt=-2n\log x=-2n\log \sqrt{1-\frac{v}{n}}=v+\mathcal{O}(n^{-1}), 
%\frac{v}{2M}+\mathcal{O}(M^{-2}), \qquad M\to\infty.
\end{equation}
as $n\to\infty$, and this leads to \eqref{eq:asympDt}.
%therefore
%\begin{equation}
%D_M(t)
%=
%\left(\frac{M}{v}\right)^{\frac{\gamma}{2}\left(\frac{\gamma}{2}+\alpha\right)}
%\exp\left(-\int_{2Mt}^{\infty}\frac{\sigma(u)}{u}du\right)(1+o(1)).
%\end{equation}
\end{proof}

Regarding the prefactor in \eqref{eq:RgammaDM}, using the known asymptotic expansion for the Barnes G function as $n\to\infty$, see for instance \cite[5.7.15]{NIST:DLMF}, we obtain
%\begin{equation}
%\begin{aligned}
%R_{\gamma}(x)
%&=
%\prod_{j=0}^{M-1}\frac{\Gamma(\frac{\gamma}{2}+j+1)\Gamma(j+\alpha+1)}{j!\Gamma(\frac{\gamma}{2}+j+\alpha+1)}
%%D_{M}
%%\prod_{j=0}^{M-1}\frac{\Gamma(\frac{\gamma}{2}+j+1)}{\Gamma(\frac{\gamma}{2}+j+\alpha+1)}
%D_{M}(t)\\
%&=
%\frac{G\left(\frac{\gamma}{2}+M+1\right)G(\alpha+M+1)}{G(M+1)G\left(\frac{\gamma}{2}+M+\alpha+1\right)}
%D_{M}(t)
%\end{aligned}
%\end{equation}
%Using the asymptotic expansion for the Barnes G function as $M\to\infty$, we have
\begin{equation}
\frac{G\left(\frac{\gamma}{2}+n+1\right)G(\alpha+n+1)}{G(n+1)G\left(\frac{\gamma}{2}+n+\alpha+1\right)}
=
n^{-\frac{\alpha\gamma}{2}}\left(1+\mathcal{O}(n^{-1})\right).
\end{equation}
Writing everything together we arrive at 
%\begin{equation}
%\begin{aligned}
%R_{\gamma}(x)
%&=
%%e^{\left(\frac{\gamma}{2}+\alpha-1\right)Mt}
%M^{\frac{\gamma^2}{4}}\left(1+\mathcal{O}(M^{-1})\right)
%v^{-\frac{\gamma}{2}\left(\frac{\gamma}{2}+\alpha\right)}
%\exp\left(-\int_{2Mt}^{\infty} \frac{\sigma(u)}{u}du\right)
%(1+o(1)).
%\end{aligned}
%\end{equation}
%as $M\to\infty$. This gives 
\eqref{eq:Rgamma_PV}. The range of parameters $a,b$ for which the result remains valid is then extended following \cite[Theorem 1.4]{CIK11}

%We note that this result is consistent with the asymptotic formula for even integer values of $\gamma$ given in Theorem \ref{th:trunc-asympt}, identifying $\gamma=2k$. % and the scaled variable 

\appendix

\section{Reduction of planar orthogonality to contours}
\label{sec:planartocontour}
\begin{proof}[Proof of Lemma \ref{lem:planartocontour}]
Let $\alpha = N-n$ and assume $x > 0$ without loss of generality. We begin with the planar orthogonality
\begin{equation}
\label{corthog1}
\frac{\delta_{jk}}{\chi_{j}} = \int_{\mathbb{D}}p_{j}(z)\overline{z}^{k}|z-x|^{\gamma}(1-|z|^{2})^{\alpha-1}d^{2}z, \qquad 0 \leq k \leq j,
\end{equation}
We can replace $\overline{z}^{k}$ with $(\overline{z}-x)^{k}$ in \eqref{corthog1} without altering the orthogonality:
\begin{equation}
\frac{\delta_{jk}}{\chi_{j}} = \int_{\mathbb{D}}p_{j}(z)(\overline{z}-x)^{k}|z-x|^{\gamma}(1-|z|^{2})^{\alpha-1}d^{2}z, \qquad 0 \leq k \leq j, \label{corthog2}
\end{equation}
which will be more convenient. Define the function
\begin{equation}
h_k(z,\overline{z},x) = \int_{x}^{\overline{z}}(s-x)^{\frac{\gamma}{2}+k}(1-z s)^{\alpha-1}ds \label{hdef}
\end{equation}
where the integration contour connects initial point $x$ and terminal point $\overline{z}$. The root is defined with respect to the principal branch, i.e. it has a branch cut on $(-\infty,x)$. Then \eqref{corthog2} becomes
\begin{equation}
\begin{split}
\frac{\delta_{jk}}{\chi_{j}} &= \int_{\mathbb{D}}p_{j}(z)(z-x)^{\frac{\gamma}{2}}\,\frac{\partial}{\partial \overline{z}}h(z,\overline{z},x)\,d^{2}z\\
&=\frac{1}{2i}\oint_{S^{1}}p_{j}(z)(z-x)^{\frac{\gamma}{2}}h_k(z,\overline{z},x)\,dz, \qquad 0 \leq k \leq j,
\end{split}
\end{equation}
where we applied Green's theorem to reduce the integral over the unit disc to a line integral around the unit circle with counter clockwise orientation (for $\gamma\in(-2,0)$ there is singular behavior near $z=x$, but one can still verify the statement. We omit the details). Now looking at the function $h_k(z,\overline{z},x)$, perform the substitution
\begin{equation}
u = \frac{s-x}{\overline{z}-x} \iff s= u(\overline{z}-x)+x.
\end{equation}
This maps the integration contour in \eqref{hdef} to the unit interval $[0,1]$. Furthermore, since $|z|^{2}=1$, we get the identity
\begin{equation}
1-zs = 1-zu(\overline{z}-x)-zx = 1-u+zux-zx = (1-u)(1-zx).
\end{equation}
So for $z \in S^{1}$, we rewrite \eqref{hdef} as 
\begin{equation}
\begin{split}
h(z,\overline{z},x) &= (\overline{z}-x)^{\frac{\gamma}{2}+k+1}(1-zx)^{\alpha-1}\int_{0}^{1}u^{\frac{\gamma}{2}+k}(1-u)^{\alpha-1}du\\
&= c_{\alpha,\gamma,k}(\overline{z}-x)^{\frac{\gamma}{2}+k+1}(1-zx)^{\alpha-1}
\end{split}
\end{equation}
where
\begin{equation}
\begin{split}
c_{\alpha,\gamma,k} &:= \int_{0}^{1}u^{\frac{\gamma}{2}+k}(1-u)^{\alpha-1}du\\
&= \frac{\Gamma(\frac{\gamma}{2}+k+1)\Gamma(\alpha)}{\Gamma(\frac{\gamma}{2}+k+\alpha+1)}
\end{split}
\end{equation}
is the Euler beta integral. Then the obtained orthogonality reads
\begin{equation}
\begin{split}
\frac{\delta_{jk}}{\chi_{j}} &= \frac{c_{\alpha,\gamma,k}}{2i}\oint_{S^{1}}p_{j}(z)(z-x)^{\frac{\gamma}{2}}(\overline{z}-x)^{\frac{\gamma}{2}+k+1}(1-zx)^{\alpha-1}\,dz\\
&= \frac{c_{\alpha,\gamma,k}}{2i}\oint_{S^{1}}p_{j}(z) |z-x|^{\gamma}\left(\frac1z-x\right)^{k}(1-zx)^{\alpha}\, \frac{dz}{z},
\end{split}
\end{equation}
for $0 \leq k \leq j$. Since for any $\ell\geq 0$ one has the trivial identity
\begin{align*}
\frac1{z^{\ell}} = \sum_{k=0}^{\ell} \binom{\ell}{k} x^{\ell-k} \left(\frac1z-x\right)^k
\end{align*}
one can take linear combinations and equivalently express the orthogonality as
\begin{equation}
\frac{\delta_{jk}}{\chi_{j}} = \frac{c_{\alpha,\gamma,j}}{2i}\oint_{S^{1}}p_{j}(z)z^{-k}|z-x|^{\gamma}(1-zx)^{\alpha}\, \frac{dz}{z}
\end{equation}
for $0 \leq k \leq j$. We write, for $z \in S^{1}$,
\begin{equation}
|z-x|^{\gamma} = (z-x)^{\frac{\gamma}{2}}(1/z-x)^{\frac{\gamma}{2}}
\end{equation}
with cuts on $(-\infty,x)$ and $(1/x,\infty)$. Scaling $z \to zx$ so that the new roots have cuts on $(-\infty,1)$ and $(1/x^{2},\infty)$, the orthogonality becomes
\begin{equation}
\frac{\delta_{jk}}{\chi_{j}} = \frac{c_{\alpha,\gamma,j}}{2i}\oint_{S^{1}/x}\tilde p_{j}(z)z^{-k}(z-1)^{\frac{\gamma}{2}}(1/z-x^{2})^{\frac{\gamma}{2}}(1-x^{2}z)^{\alpha}\, \frac{dz}{z}
\end{equation}
where we used that the factor $x^{j-k}$ can be replaced with $1$. Then we get the advertised weight after noting that $(1/z-x^{2})^{\frac{\gamma}{2}} = z^{-\frac{\gamma}{2}}(1-x^{2}z)^{\frac{\gamma}{2}}$ with a cut on $[-\infty,0)$. This combines with the previous cut to give a cut on $(0,1)$ and a cut on $(1/x^{2},\infty)$. Since $x \in [0,1]$, the contour $S^{1}/x$ goes between the two branch cuts and we can shrink it back to $S^{1}$ by analyticity.
\end{proof}

\section{Differential identity and partition function}
\label{app:diffid}
In this appendix we give the proof of Lemma \ref{lem:diffid}. For convenience we restate it here. Consider the partition function in \eqref{momsgin}, which we denote
\begin{equation}
R_{\gamma}(x) = \mathbb{E}\left(|\det(B_{n}-x)|^{\gamma}\right).
\end{equation}
\begin{lemma}
The following differential identity holds:
\begin{equation}
\begin{split}
\frac{d}{dx}(\log R_{\gamma}(x)) &= -2x\left(\frac{\gamma}{2}+\alpha\right)\,\left( -nu^{n+1}q_{n}(u^{-1})I_{12}(u)\right.\\
&\left.-nu+u^{3}[\partial_{z}\tilde{p}_{n}(z)]_{z=u}\,I_{22}(u)\right.\\
&\left.-u^{n+2}[\partial_{z}q_{n}(z^{-1})]_{z=u}\,I_{12}(u)\right) \label{app:diffid}
\end{split}
\end{equation}
where $u=x^{-2}$ and
\begin{equation}\label{eq:I12I22}
\begin{split}
I_{12}(u) &= \oint_{\Sigma}\frac{z^{-n+1}\tilde{p}_{n}(z)}{u-z}\frac{w(z)}{2iz}dz\\
I_{22}(u) &= \oint_{\Sigma}\frac{q_{n}(z^{-1})}{u-z}\frac{w(z)}{2iz}dz.
\end{split}
\end{equation}
\end{lemma}

\begin{proof}
Starting with identity \eqref{momsgin}, we write
\begin{equation}
\begin{split}
\frac{d}{dx}(\log R_{\gamma}(x)) &= -2\sum_{j=0}^{n-1}\frac{\partial_{x}\chi_{j}}{\chi_{j}}. \label{dgam}
\end{split}
\end{equation}
We will make use of the following pair of identities
\begin{align}
\oint_{\Sigma}[\partial_{x} \tilde{p}_{j}(z)]q_{j}(z^{-1})w(z)\,\frac{dz}{2iz} &= \frac{\partial_{x} \chi_{j}}{\chi_{j}}, \\
\oint_{\Sigma}\tilde{p}_{j}(z)[\partial_{x} q_{j}(z^{-1})]w(z)\,\frac{dz}{2iz} &= \frac{\partial_{x} \hat{\chi}_{j}}{\hat{\chi}_{j}}.
\end{align}
These follow immediately from taking the leading terms, of $\tilde{p}_{j}(z)$ in the first identity and of $q_{j}(z^{-1})$ in the second and applying the orthogonality relations \eqref{constr}, \eqref{qorthog} and \eqref{biorth}. Furthermore note that
\begin{equation}
\frac{\partial_{x} \chi_{j}}{\chi_{j}} = \frac{\partial_{x} \hat{\chi}_{j}}{\hat{\chi}_{j}},
\end{equation}
which follows directly from \eqref{norm-rel}. From the above identities, we get
\begin{equation}
-2\frac{\partial_{x}\chi_{j}}{\chi_{j}} = -\oint_{\Sigma}[\partial_{x}(\tilde{p}_{j}(z)q_{j}(z^{-1}))]w(z)\,\frac{dz}{2iz}.
\end{equation}
Then \eqref{dgam} becomes
\begin{equation}
\frac{d}{dx}(\log R_{\gamma}(x)) = -\sum_{j=0}^{n-1}\oint_{\Sigma}[\partial_{x}(\tilde{p}_{j}(z)q_{j}(z^{-1}))]w(z)\,\frac{dz}{2iz}. \label{dgam2}
\end{equation}
Since \eqref{biorth} is constant in $x$, we may pass the derivative in \eqref{dgam2} onto the weight and obtain
\begin{equation}
\begin{split}
\frac{d}{dx}(\log R_{\gamma}(x)) &= \sum_{j=0}^{n-1}\oint_{\Sigma}\tilde{p}_{j}(z)q_{j}(z^{-1})[\partial_{x} w(z)]\,\frac{dz}{2iz}\\
&= -2x\left(\frac{\gamma}{2}+\alpha\right)\sum_{j=0}^{n-1}\oint_{\Sigma}\frac{z\tilde{p}_{j}(z)q_{j}(z^{-1})}{1-x^{2}z}\,\frac{w(z)dz}{2iz}. \label{log-deriv}
\end{split}
\end{equation}
In obtaining \eqref{log-deriv}, we used that
\begin{equation}
\partial_{x}w(z) = \frac{-2xz(\frac{\gamma}{2}+\alpha)}{1-x^{2}z}\,w(z)
\end{equation}
which follows from the definition of the weight in \eqref{cweight}. Next we perform the summation in \eqref{log-deriv}. The Christoffel-Darboux formula for bi-orthogonal polynomials (see e.g. Appendix A in \cite{WW19}) tells us that
\begin{equation}
\begin{split}
&\sum_{j=0}^{n-1}\tilde{p}_{j}(z)q_{j}(z^{-1}) = -n\tilde{p}_{n}(z)q_{n}(z^{-1})\\
&+z(q_{n}(z^{-1})\partial_{z}\tilde{p}_{n}(z)-\tilde{p}_{n}(z)\partial_{z}q_{n}(z^{-1})).
\end{split}
\end{equation}
Inserting this into \eqref{log-deriv} results in three integrals to evaluate, which we write as
\begin{equation}
\frac{d}{dx}(\log R_{\gamma}(x)) = -2x\left(\frac{\gamma}{2}+\alpha-1\right)\,(I_{1}+I_{2}+I_{3}). \label{ijformula}
\end{equation}
In what follows we shall see that each integral $I_{j}$ gives rise to the main contribution on line $j$ of the identity \eqref{diffid}, for $j=1,2,3$. The first is
\begin{equation}
\begin{split}
I_{1} &= -n\oint_{\Sigma}\frac{z\tilde{p}_{n}(z)q_{n}(z^{-1})}{1-x^{2}z}\,\frac{w(z)dz}{2iz}\\
&= \frac{n}{x^{2}}\oint_{\Sigma}\tilde{p}_{n}(z)q_{n}(z^{-1})\,\frac{w(z)dz}{2iz}-\frac{n}{x^{2}}\oint_{\Sigma}\frac{\tilde{p}_{n}(z)q_{n}(z^{-1})}{1-x^{2}z}\,\frac{w(z)dz}{2iz}\\
&=un-u^{2}n\oint_{\Sigma}\frac{\tilde{p}_{n}(z)q_{n}(z^{-1})}{u-z}\,\frac{w(z)dz}{2iz},
\end{split}
\end{equation}
where in the first line we wrote $z = -\frac{1}{x^{2}}(1-x^{2}z)+\frac{1}{x^{2}}$ and the last line follows from orthogonality. Then we write
\begin{equation}
\begin{split}
\label{qmdiv}
&\frac{q_{n}(z^{-1})}{u-z} = \frac{z^{-n+1}z^{n-1}q_{n}(z^{-1})}{u-z}\\
&= u^{n-1}q_{n}(u^{-1})\frac{z^{-n+1}}{u-z} + z^{-n+1}\,\frac{z^{n-1}q_{n}(z^{-1})-u^{n-1}q_{n}(u^{-1})}{u-z}
\end{split}
\end{equation}
and notice that 
\begin{equation}
z^{-n+1}\,\frac{z^{n-1}q_{n}(z^{-1})-u^{n-1}q_{n}(u^{-1})}{u-z} = \frac{\hat{\chi}_{n}}{u}z^{-n}+\mathcal{P}_{n-1}(z^{-1})
\end{equation}
where $\mathcal{P}_{n-1}$ is some degree $n-1$ polynomial, which does not contribute because of orthogonality. The contribution of the leading term above cancels the first one in $I_{1}$ and we find that 
\begin{equation}
I_{1} = -n u^{n+1}q_{n}(u^{-1})\oint_{\Sigma}\frac{\tilde{p}_{n}(z)z^{-n+1}}{u-z}\,\frac{w(z)}{2iz}dz. \label{i1final}
\end{equation}
The next integral is 
\begin{equation}
I_{2} = u\oint_{\Sigma}\frac{z^{2}q_{n}(z^{-1})[\partial_{z}\tilde{p}_{n}(z)]}{u-z}\frac{w(z)}{2iz}dz.
\end{equation}
We write $z^{2} = -z(u-z)+uz$ so that
\begin{equation}
\begin{split}
I_{2} &= -u\oint_{\Sigma}zq_{n}(z^{-1})[\partial_{z}\tilde{p}_{n}(z)]\frac{w(z)}{2iz}dz\\
&+ u^{2}\oint_{\Sigma}\frac{zq_{n}(z^{-1})[\partial_{z}\tilde{p}_{n}(z)]}{u-z}\frac{w(z)}{2iz}dz
\end{split}
\end{equation}
and by orthogonality
\begin{equation}
I_{2} = -n u + u^{2}\oint_{\Sigma}\frac{q_{n}(z^{-1})z[\partial_{z}\tilde{p}_{n}(z)]}{u-z}\frac{w(z)}{2iz}dz.
\end{equation}
Now similarly, write
\begin{equation}
\begin{split}
\frac{z[\partial_{z}\tilde{p}_{n}(z)]}{u-z} &= \frac{u[\partial_{u}\tilde{p}_{n}(u)]}{u-z}+\frac{z[\partial_{z}\tilde{p}_{n}(z)]-u[\partial_{u}\tilde{p}_{n}(u)]}{u-z}\\
&=\frac{u[\partial_{u}\tilde{p}_{n}(u)]}{u-z}+\mathcal{P}_{n-1}(z)
\end{split}
\end{equation}
where $\mathcal{P}_{n-1}(z)$ is some degree $n-1$ polynomial. By orthogonality, only the first term above can contribute and so
\begin{equation}
I_{2} = -n u+u^{3}[\partial_{u}\tilde{p}_{n}(u)]\oint_{\Sigma}\frac{q_{n}(z^{-1})}{u-z}\frac{w(z)}{2iz}dz. \label{i2final}
\end{equation}
We have
\begin{equation}
I_{3} = -u\int_{\Sigma}\frac{z^{2}\tilde{p}_{n}(z)[\partial_{z}q_{n}(z^{-1})]}{u-z}\frac{w(z)}{2iz}dz
\end{equation}
Writing $z^{2} = -(u-z)z+uz$ we obtain $I_{3} = I_{3,1} + I_{3,2}$, where
\begin{equation}
\begin{split}
I_{3,1} &=  u\int_{\Sigma}\tilde{p}_{n}(z)z[\partial_{z}q_{n}(z^{-1})]\frac{w(z)}{2iz}dz\\
I_{3,2} &= -u^{2}\int_{\Sigma}\frac{\tilde{p}_{n}(z)z[\partial_{z}q_{n}(z^{-1})]}{u-z}\frac{w(z)}{2iz}dz.
\end{split}
\end{equation}
We have $I_{3,1} = -n u$ and analogously to the steps in \eqref{qmdiv}, we obtain
\begin{equation}
I_{3,2} = n u-u^{n+2}[\partial_{u}q_{n}(u^{-1})]\int_{\Sigma}\frac{z^{-n+1}\tilde{p}_{n}(z)}{u-z}\frac{w(z)}{2iz}dz
\end{equation}
so that
\begin{equation}
I_{3} = -u^{n+2}[\partial_{u}q_{n}(u^{-1})]\int_{\Sigma}\frac{z^{-n+1}\tilde{p}_{n}(z)}{u-z}\frac{w(z)}{2iz}dz. \label{i3final}
\end{equation}
Inserting \eqref{i1final}, \eqref{i2final} and \eqref{i3final} into \eqref{ijformula} completes the proof.
\end{proof}

\section{The inverse image of $\phi_r$} \label{Apx:phiAnalyticIso}

In this section we prove that $\phi_r=\phi-\ell$, as defined in \eqref{eq:phir}, is an analytic isomorphism when restricted to certain sets. Equivalently, we may instead investigate $e^\phi$, and we may then ignore the cut at $(-\infty,0]$. For any set $S\subset \mathbb C$ we denote the set of its complex conjugate elements by $S^*$. 

\begin{lemma} \label{lem:phiAnalyticIso}
Let $c\in(0,1)$.
There exist pairwise disjoint open sets $\mathcal S_1\subset\mathbb C$ and $\mathcal S_2, \mathcal S_3\subset\mathbb H$ such that $e^\phi$ is an analytic isomorphism from:
\begin{itemize}
\item[(i)] $\mathcal S_1$ to $\{\zeta\in\mathbb C\setminus\{0\} : \arg \zeta\in (-\pi,-\pi c)\cup (\pi c,\pi]\}$.
\item[(ii)] $\mathcal S_2$ and $\mathcal S_3^*$ to $\{\zeta\in\mathbb C\setminus\{0\} : \arg \zeta\in (0,\pi c)\}$
\item[(iii)] $\mathcal S_2^*$ and $\mathcal S_3$ to $\{\zeta\in\mathbb C\setminus\{0\} : \arg \zeta\in (-\pi c,0)\}$.
\end{itemize}
We have $(-\infty,0)\subset \mathcal S_1, (0,z_0)\subset \partial \mathcal S_2, (z_0,\infty)\subset \partial \mathcal S_3$, and $\arg z$ is arbitrarily close to values in the interval $[\frac{\pi c}{c+1}, \frac{2\pi c}{c+1}]$ as $|z|\to\infty$ in $\mathcal S_2$ (see Figure \ref{fig:S1S2S3}). 
\end{lemma}

\begin{proof}
Let $\zeta\in\mathbb C\setminus\{0\}$ and impose $\arg \zeta\neq \pm \pi c$. Defining $w=x^{-2}-z$, the equation $e^{\phi(z)}=\zeta$ is equivalent to
\begin{align} \label{eq:solutionForGammarw}
f(w):= w^{c+1}-x^{-2} w^c = -x^{-2c} \zeta,
\end{align}
where the cut for $w^c$ and $w^{c+1}$ has to be placed at $(-\infty,0]$. We will find out how many solutions $w\in \mathbb C\setminus (-\infty,0]$ (or equivalently $z\in\mathbb C\setminus [x^{-2},\infty)$) to  \eqref{eq:solutionForGammarw} exist. Consider the contour $\gamma$ one gets by taking $|w|=R$, and replacing a part by a key hole contour around $(-R,0]$, to create a closed contour. The number of solutions to \eqref{eq:solutionForGammarw} in the interior of $\gamma$ is
\begin{align} \label{eq:numberZeros}
\frac1{2\pi i} \oint_\gamma \frac{f'(w)}{f(w)+x^{-2c}\zeta} dw.
\end{align}
Collapsing the horizontal line segments to the negative real line and verifying that the contribution near $w=0$ is negligible, we obtain that the number of solutions to \eqref{eq:solutionForGammarw} in $|w|<R$ (excluding $w\in [-R,0]$) is given by
\begin{multline*}
\frac1{2\pi i}\oint_{|w|=R} \frac{(c+1)w^c-c x^{-2}w^{c-1}}{w^{c+1}-x^{-2} w^c+x^{-2c}\zeta} dw 
+ \frac1{2\pi i}\int_0^R \frac{(c+1)t^c+c x^{-2} t^{c-1}}{t^{c+1}+x^{-2} t^c-x^{-2c} e^{\pi i c}\zeta} dt\\
- \frac1{2\pi i}\int_0^R \frac{(c+1)t^c+c x^{-2} t^{c-1}}{t^{c+1}+x^{-2} t^c-x^{-2c} e^{-\pi i c}\zeta} dt.
\end{multline*}
Since $\arg \zeta\neq \pm \pi c$ this expression is well-defined and it equals
\begin{align*}
&\frac1{2\pi i}\oint_{|w|=R}  \frac{(c+1)w^c-c x^{-2}w^{c-1}}{w^{c+1}-x^{-2} w^c+x^{-2c}\zeta} dw 
-\arg(-e^{\pi i c} \zeta)+\arg(-e^{-\pi i c}\zeta)\\
&+\frac1{2\pi i}\log (R^{c+1}+x^{-2}R^c-x^{-2c} e^{\pi i c}\zeta)
-\frac1{2\pi i}\log(R^{c+1}+x^{-2}R^c-x^{-2c} e^{-\pi ic}\zeta),
\end{align*}
where the logarithms (and arguments) have cut $(-\infty,0]$.  Now we use that 
\begin{align*}
\arg(-e^{\pi i c} \zeta)-\arg(-e^{-\pi i c}\zeta)=2\pi
\begin{cases}
c, & \arg\zeta\in (-\pi,-\pi c)\cup (\pi c,\pi],\\
c-1, & \arg\zeta\in (-\pi c, \pi c).
\end{cases}
\end{align*}
Thus \eqref{eq:numberZeros} has limit $1$ as $R\to\infty$ when $|\arg \zeta|\in(\pi c,\pi]$ and limit $2$ when $|\arg \zeta|\in[0,\pi c)$. We conclude that \eqref{eq:solutionForGammarw} has a unique solution for each $\zeta\in \mathbb C$ with $\arg\zeta\in(-\pi,-\pi c)\cup (\pi c, \pi]$. $\mathcal S_1$ is defined as the collection of these solutions. Clearly, $(-\infty,0)\subset \mathcal S_1$.  

\begin{figure}[h!]
\centerline{\includegraphics[scale=1]{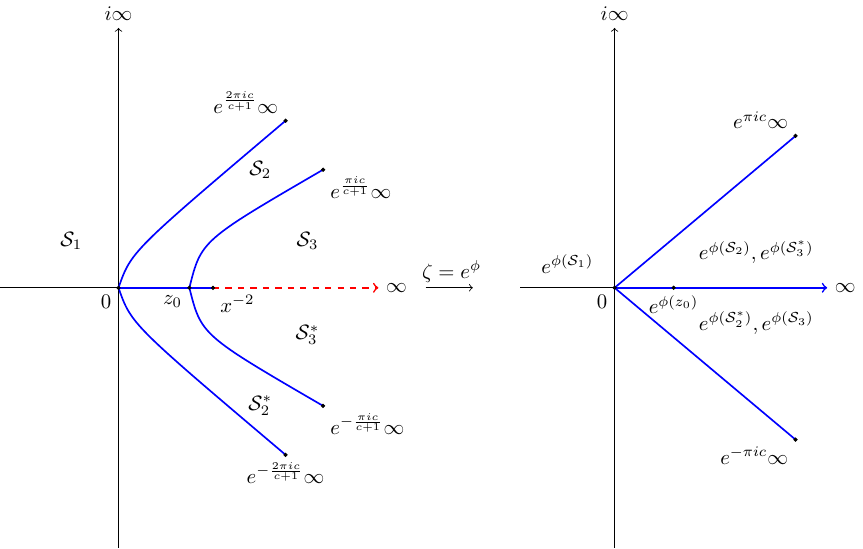}}
\caption{Schematic representation of the sets $\mathcal S_1, \mathcal S_2, \mathcal S_2^*, \mathcal S_3$ and $\mathcal S_3^*$ from Lemma \ref{lem:phiAnalyticIso}. $\partial \mathcal S_2\cap \partial \mathcal S_3$ and $\partial \mathcal S_2^*\cap \partial \mathcal S_3^*$ are mapped to $[e^{\phi(z_0)}, \infty)$. $[0,z_0]$ and $[z_0,x^{-2}]$ are mapped to $[0,e^{\phi(z_0)}]$. $\partial \mathcal S_1$ and $[x^{-2},\infty)$ correspond to $\arg \zeta=\pm \pi c$.}
\label{fig:S1S2S3}
\end{figure}

For cases (ii) and (iii), note that $e^{\phi}$ maps both $(0,z_0)$ and $(z_0,x^{-2})$ bijectively to $(0,e^{\phi(z_0)})$. Hence, the inverse image of $(-\infty,e^{\phi(z_0)}]$ under $e^\phi$ is $(-\infty,x^{-2}]$. When $\zeta>e^{\phi(z_0)}$, solutions $z=x^{-2}-w$ to \eqref{eq:solutionForGammarw} cannot be real and necessarily come in complex conjugate pairs, thus the inverse image of $[e^{\phi(z_0)},\infty)$ under $e^\phi$ is a curve $\gamma_{z_0}$ which is symmetric with respect to the horizontal axis. The curves $\gamma_{z_0}$ and $[0,\infty)$ divide $\mathbb C\setminus \overline{\mathcal S_1}$ into four parts $\mathcal S_2, \mathcal S_3\subset \mathbb H$ and their complex conjugate sets. We define $\mathcal S_2$ as the component such that $[0,z_0]\subset \partial S_2$ and $\mathcal S_3$ as the component such that $[z_0,\infty)\subset \partial S_3$. Inspecting the boundaries of the four regions, and noting that $e^{\phi(\overline z)}=\overline{e^{\phi(z)}}$ we infer that the inverse image of $\{\zeta\in\mathbb C\setminus\{0\} : \arg \zeta\in (0,\pi c)\}$ is the union of two of $\mathcal S_1, \mathcal S_2, \mathcal S_1^*$ and $\mathcal S_2^*$ and the inverse image of $\{\zeta\in\mathbb C\setminus\{0\} : \arg \zeta\in (-\pi c,0)\}$ is the union of the remaining two. For any fixed $\sigma\in(0,z_0)\cup (z_0,x^{-2})$ we see that $\operatorname{Im} e^{\phi(\sigma+i\epsilon)} = (1-x^2\sigma)^{c-1} (1-\sigma/z_0) \epsilon+\mathcal O(\epsilon^2)$ for small $\epsilon$, hence $\mathcal S_2$ is mapped to $\{\zeta\in\mathbb C\setminus\{0\} : \arg \zeta\in (0,\pi c)\}$ while $\mathcal S_3$ is mapped to $\{\zeta\in\mathbb C\setminus\{0\} : \arg \zeta\in (0,\pi c)\}$.

The final statement concerning the boundary of $\mathcal S_2$ follows from the fact that $e^{\phi(z)}=x^{-2c}e^{-\pi i c} z^{c+1}(1+\mathcal O(1/z))$ for large $z$ in $\mathcal S_2$. Thus, solving \eqref{eq:solutionForGammarw} for $\arg \zeta=0$ and $\arg \zeta$ slightly smaller than $\pi c$, we obtain the statement. 
\end{proof}

\bibliographystyle{plain}
\bibliography{bibliography}

\end{document}